\documentclass[accepted=2022-03-20,a4paper,amsfonts,amssymb,floatfix]{quantumarticle}
\pdfoutput=1
\usepackage[utf8]{inputenc}
\usepackage[english]{babel}
\usepackage[T1]{fontenc}
\PassOptionsToPackage{compress}{natbib}
\usepackage[numbers]{natbib}
\usepackage{amsmath}				
\usepackage{mathtools}						
\usepackage{amsthm}						
\usepackage{subcaption}						
\usepackage{dsfont}						
\usepackage{tikz}						
\usetikzlibrary{calc,shapes,graphs,positioning}			
\usepackage{microtype}						

\newtheorem{theorem}{Theorem}
\newtheorem{corollary}{Corollary}
\newtheorem{lemma}{Lemma}
\theoremstyle{definition}
\newtheorem{definition}{Definition}
\newtheorem*{game*}{Game}
\newtheorem*{asgame*}{Ardehali-Svetlichny Game}

\DeclareMathOperator{\Tr}{Tr}				
\newcommand{\ket}[1]{|#1\rangle}			
\newcommand{\bra}[1]{\langle #1|}			
\newcommand{\doubleket}[1]{|#1\rangle\rangle}		
\newcommand{\doublebra}[1]{\langle\langle #1|}		
\newcommand{\id}{\mathds{1}}				
\newcommand{\cc}[1]{\operatorname{\mathsf{#1}}}         
\newcommand\party[2]{ 
	\coordinate (#1L) at ($ (#1) - (1cm,0) $);
	\coordinate (#1R) at ($ (#1) + (1cm,0) $);
	\draw[-,fill=blue!20!white] (#1L) to[out=#2,in=180-#2] node[midway] (#1T) {} (#1R) to[out=#2-180,in=-#2] node[midway] (#1B) {} (#1L);
}
\newcommand\knob[3]{ 
	\node at ($ (#1) + (-45:0.3) $) {\tiny #3};
	\draw (#1) circle [radius=0.1];
	\draw (#1) circle [radius=0.15];
	\foreach \x in {1,...,12}
	{
		\def\angle{\x*360/12}
		\draw[thick] ($ (#1) + (\angle:0.1) $) -- +(\angle:0.05);
	}
	\def\angle{#2*360/12}
	\draw (#1) -- +(\angle:0.1);
}
\newcommand\meter[3]{ 
	\node at ($ (#1) + (180-45:0.3) $) {\tiny #3};
	\draw[thick] ($ (#1) - (0.15,0.1) $) rectangle ($ (#1) + (0.15,0.1) $);
	\draw ($ (#1) - (0,0.05) $) ++(20:0.1) arc (20:160:0.1);
	\def\angle{#2}
	\draw ($ (#1) - (0,0.05) $) -- ++(\angle:0.15);
}

\DeclarePairedDelimiter{\ceil}{\lceil}{\rceil}
\DeclarePairedDelimiter{\floor}{\lfloor}{\rfloor}

\begin{document}
\title{Unlimited non-causal correlations and their relation to non-locality}
\author{\"Amin Baumeler}
\orcid{0000-0001-6760-0439}
\affiliation{Institute for Quantum Optics and Quantum Information (IQOQI-Vienna), Austrian Academy of Sciences, 1090 Vienna, Austria}
\affiliation{Faculty of Physics, University of Vienna, 1090 Vienna, Austria}
\affiliation{Facolt\`{a} indipendente di Gandria, 6978 Gandria, Switzerland}
\author{Amin Shiraz Gilani}
\affiliation{Institute for Quantum Optics and Quantum Information (IQOQI-Vienna), Austrian Academy of Sciences, 1090 Vienna, Austria}
\affiliation{Department of Computer Science, University of Maryland, College Park, Maryland 20742, USA}
\author{Jibran Rashid}
\orcid{0000-0002-6927-7417}
\affiliation{School of Mathematics and Computer Science, Institute of Business Administration, Karachi, Pakistan}

\begin{abstract}
	\noindent
	Non-causal correlations certify the lack of a~definite causal order among localized space-time regions.
	In stark contrast to scenarios where a single region influences its own causal past,
	some processes that distribute non-causal correlations satisfy a series of natural desiderata: logical consistency, linear and reversible dynamics, and computational tameness.
	Here, we present such processes among arbitrary many regions where each region influences every other but itself,
	and show that the above desiderata are altogether {\em insufficient to limit the amount of ``acausality''\/} of non-causal correlations.
	This leaves open the identification of a principle that forbids non-causal correlations.
	Our results exhibit {\em qualitative and quantitative parallels\/} with the non-local correlations due to Ardehali and~Svetlichny.
\end{abstract}

\maketitle

The succession of events is usually assumed to follow a {\em fixed causal order:\/}
The causal structure~is directed acyclic~\cite{Pearl2009}.
General relativity describes {\em dynamic causal order,\/} and according to quantum theory physical quantities are {\em indefinite.\/}
Therefore, it is reasonable to expect that a satisfactory theory of quantum gravity exhibits {\em both\/} features~\cite{Hardy2005}.
This is exemplarily demonstrated by the quantum switch~\cite{Chiribella2013}:
A quantum system coherently controls the causal order between two events in its future~\cite{Colnaghi2012,Zych2017}.
It is known~\mbox{\cite{multipartyprocesses,Araujo2015,Purves2021,Wechs2021}} that the quantum switch and generalizations thereof do not violate {\em causal inequalities.\/}
Causal inequalities~\cite{Oreshkov2012,simplestcausalinequalities}, similar to Bell inequalities~\cite{Bell1964}, are theory independent and confine the observable correlations among a set of agents under the assumption of a {\em definite causal order.\/}
Although this assumption is natural, there exist motivations to study the world beyond.
In a world beyond, for instance, one can ask:
{\em How can we derive causal order without presupposing causal order, and what is the logical origin of causal order?}\footnote{This question is similar to Wheeler's puzzle: {\em ``How to derive time without presupposing time''}~\cite{Wheeler1988}.}
This question is of foundational interest and relevant to general relativity and quantum gravity~\cite{Hawking1992}.
In general relativity, for instance, no causal order is enforced, and Einstein's suspicion~\cite{EinsteinCTC} that closed time-like curves are consistent with that theory proved true~\cite{Lanczos1924,Godel1949}.
Thus, this question asks for a principle with which such exotic space-time structures are excluded.
\begin{figure}
	\centering
	\begin{tikzpicture}
		\def\knobx{-0.6}
		\def\knoby{0}
		\def\alpha{45}
		\def\kA{2}
		\def\kB{10}
		\def\kC{5}
		\def\mA{140}
		\def\mB{80}
		\def\mC{60}
		\coordinate (A) at (0,0);
		\party{A}{\alpha}
		\node at (A) {\tiny $\mu^{a,x}$};
		\coordinate (AK) at ($ (A) + (\knobx,\knoby) $);
		\knob{AK}{\kA}{$X$}
		\coordinate (AM) at ($ (A) + (-\knobx,-\knoby) $);
		\meter{AM}{\mA}{$A$}
		\coordinate (B) at ($ (AR) + (1.5cm, 0) $);
		\party{B}{\alpha}
		\node at (B) {\tiny $\nu^{b,y}$};
		\coordinate (BK) at ($ (B) + (\knobx,\knoby) $);
		\knob{BK}{\kB}{$Y$}
		\coordinate (BM) at ($ (B) + (-\knobx,-\knoby) $);
		\meter{BM}{\mB}{$B$}
		\coordinate (C) at ($ (BR) + (1.5cm, 0) $);
		\party{C}{\alpha}
		\node at (C) {\tiny $\tau^{c,z}$};
		\coordinate (CK) at ($ (C) + (\knobx,\knoby) $);
		\knob{CK}{\kC}{$Z$}
		\coordinate (CM) at ($ (C) + (-\knobx,-\knoby) $);
		\meter{CM}{\mC}{$C$}
		\def\offset{0.2}
		\def\offangle{5}
		\def\extrax{0.4}
		\def\extray{1}
		\def\c{violet!10!white}
		\path[draw,\c,fill=\c] ($ (AL) + (-\extrax,\extray) $) -- ($ (AL) + (-\extrax,0) $)
		-- ($ (AL) - (\offset,0) $) to[out=\alpha+\offangle,in=180-\alpha-\offangle] ($ (AR) + (\offset,0) $)
		-- ($ (BL) - (\offset,0) $) to[out=\alpha+\offangle,in=180-\alpha-\offangle] ($ (BR) + (\offset,0) $)
		-- ($ (CL) - (\offset,0) $) to[out=\alpha+\offangle,in=180-\alpha-\offangle] ($ (CR) + (\offset,0) $)
		-- ($ (CR) + (\extrax,0) $) -- ($ (CR) + (\extrax,\extray) $)
		-- cycle;
		\path[draw,\c,fill=\c] ($ (AL) + (-\extrax,-\extray) $) -- ($ (AL) + (-\extrax,0) $)
		-- ($ (AL) - (\offset,0) $) to[out=-\alpha-\offangle,in=180+\alpha+\offangle] ($ (AR) + (\offset,0) $)
		-- ($ (BL) - (\offset,0) $) to[out=-\alpha-\offangle,in=180+\alpha+\offangle] ($ (BR) + (\offset,0) $)
		-- ($ (CL) - (\offset,0) $) to[out=-\alpha-\offangle,in=180+\alpha+\offangle] ($ (CR) + (\offset,0) $)
		-- ($ (CR) + (\extrax,0) $) -- ($ (CR) + (\extrax,-\extray) $)
		-- cycle;
		\draw[>=stealth,double,->] (AT.center) -- ++(0,0.15) node[above,left] {\small $\mathcal O_A$};
		\draw[>=stealth,double,<-] (AB.center) -- ++(0,-0.15) node[below,left] {\small $\mathcal I_A$};
		\draw[>=stealth,double,->] (BT.center) -- ++(0,0.15) node[above,left] {\small $\mathcal O_B$};
		\draw[>=stealth,double,<-] (BB.center) -- ++(0,-0.15) node[below,left] {\small $\mathcal I_B$};
		\draw[>=stealth,double,->] (CT.center) -- ++(0,0.15) node[above,left] {\small $\mathcal O_C$};
		\draw[>=stealth,double,<-] (CB.center) -- ++(0,-0.15) node[below,left] {\small $\mathcal I_C$};
	\end{tikzpicture}
	\caption{
		In each region an experiment on a system provided by the environment is performed~(the knobs illustrate the settings and the meters the results).
		After the experiment, systems are released back to the environment.
		Process matrices~\cite{Oreshkov2012} describe the most general dynamics (functions from quantum instruments~$\mu^a_x, \nu^b_y, \tau^c_z$, which describe the experiments, to behaviors~$P_{A,B,C\mid X,Y,Z}$) such that locally no deviation from quantum theory is observed.
		Some process matrices violate causal order.
	}
	\label{fig:cc}
\end{figure}
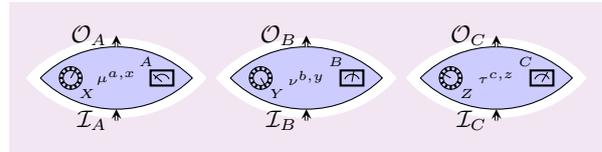

The recent process-matrix framework~\cite{Oreshkov2012} describes such a ``non-causal'' world, and relates to quantum theory in the same way as general relativity relates to special relativity:
While special relativity holds in sufficiently small space-time regions of general relativity, the process-matrix framework postulates that physics in a discrete number of local regions is described by quantum theory---and no causal order among the regions is enforced~(see Figure~\ref{fig:cc}).
This framework is known to describe indefinite causal order---{\it e.g.,\/}~the quantum switch~\cite{multipartyprocesses}---,
and moreover, {\em violates\/} causal inequalities~\cite{Oreshkov2012,Baumeler3parties,ClassicalNC1,simplestcausalinequalities,Abbott2016,Baumeler2016,multipartyprocesses,Abbott2017,Giarmatzi2019,Gu2020}.
In fact, this latter quality is {\em independent of quantum theory.\/}
Causal inequalities are also violated in the classical-probabilistic~\cite{Baumeler2016} and classical-deterministic limit~\cite{Baumeler2016fp}.
The process-matrix framework and its classical limits describe linear dynamics and comply with various desiderata:
The restriction to reversible (unitary) dynamics does not reestablish causal order~\cite{Baumeler2016fp,Araujo2017purification,Baumelerphd}, and the computational capabilities seem highly restricted~\cite{Araujo2017,Baumeler2018,Renner2021}.
In stark contrast, alternative models of violations of causal order~\cite{Deutsch1991,Hartle1994,Svetlichny2009,Svetlichny2011,Lloyd2011,Allen2014} lead to non-linear dynamics and bare unnatural features, {\it e.g.,}~quantum-state cloning~\cite{Ahn2013,Brun2013}, and extravagant computational power~\cite{Aaronson2005PostBQBPP,Aaronson2005,Aaronson2009,Aaronson2016}.

Here, we further investigate the process-matrix framework, and ask whether violations of causal inequalities vanish by increasing the number of regions---as suggested by previous studies~\cite{ClassicalNC1,Abbott2017,Araujo2017}.
The analogous question had been asked~\cite{Svetlichny1987} for violations of Bell inequalities, with the result that quantum non-local correlations are {\em unlimited:\/} They are non-vanishing for any number of bodies~\cite{Collins2002,Seevinck2002}.
We show that this is also the case here:
For any number~$n\geq 3$ of regions, {\em robust\/} violations of causal inequalities are theoretically possible, and, in contrast to the previous results, the degree of the violation increases with the number of regions.
More concretely, we design a bi-causal game~$G_n$ for~$n$ parties that is asymptotically the hardest:
As~$n$ becomes increasingly large, the maximal winning probability of the game~$G_n$ approaches~$1/2$.
Then, we show that this game is won {\em deterministically\/} in the classical-deterministic limit of the process-matrix framework.
Finally, we prove that every classical-deterministic process is a process matrix:
The game~$G_n$ is won deterministically with unitarily extendible~\cite{Araujo2017purification} process matrices.
These main findings are compactly illustrated in Figure~\ref{fig:pt}.
\begin{figure}
	\centering
	\begin{tikzpicture}
		\def\c{violet!20!white}
		\def\d{blue!20!white}
		\draw[color=\d, fill=\d] (0,0) circle (2cm);	
		\node[regular polygon, regular polygon sides=7, minimum size=4cm, fill=violet!10] (cp) {}; 
		\path[fill=red!20] (0,0) -- (cp.corner 3) -- (cp.corner 4) -- (cp.corner 5) -- (cp.corner 6) -- cycle; 
		\path[fill=black,opacity=.1] (0,.5) -- (cp.corner 3) |- (0,-2) -| (cp.corner 6) node[near end,right] {$\mathcal C^n_\textnormal{bi-causal}$} -- cycle; 
		\draw[loosely dashed] (cp.corner 6) -- (cp.corner 7) -- (cp.corner 1) -- (cp.corner 2) -- (cp.corner 3);
		\draw[dashdotted] (cp.corner 3) -- (cp.corner 4) -- (cp.corner 5) -- (cp.corner 6);
		\draw[loosely dotted] (cp.corner 3) -- (0,0) -- (cp.corner 6);
		\path[draw] (0,.5) -- (cp.corner 3) |- (0,-2) -| (cp.corner 6) node[near end,right] {$\mathcal C^n_\textnormal{bi-causal}$} -- cycle; 
		\node at (0,1) {$\mathcal C^n_\textnormal{cprocess}$};
		\node at (0,-.5) {$\mathcal C^n_\textnormal{causal}$};
		\node[anchor=east] at ($ (0,0) + (150:2cm) $) {$\mathcal C^n_\textnormal{qprocess}$};
		\fill (cp.corner 1) circle (1.75pt) node[above] {$G_n$};
		\draw[densely dotted,red] (-2.5,.5) -- (2.5,.5) node[right,align=left,font=\scriptsize,color=black] {Bi-causal\\inequality};
		\draw[densely dotted,red] (-2.5,2) -- (2.5,2);
		\path[black,latex-latex] (1.5,2cm) edge node[above right] {$\scriptstyle{ \frac{1}{2}-2^{-\ceil{n/2}}}$} (1.5,.5);
	\end{tikzpicture}
	\caption{Schematic representation of the correlation sets for~$n$ parties.
		Theorem~\ref{thm:bicausalinequality} gives the separation between the game~$G_n$ and the polytope of bi-causal correlations~$\mathcal C^n_\textnormal{bi-causal}$.
		Since the set of causal correlations is contained within the bi-causal set, this separation also holds for causal correlations.
		Theorem~\ref{thm:faces} establishes the fact that bi-causal correlations saturate the bound.
		Theorem~\ref{thm:lc} shows that the set of process-function correlations contains the point~$G_n$.
		Finally, Theorem~\ref{thm:qc} shows that every process function is a process matrix, and therefore~$\mathcal C^n_\textnormal{cprocess}\subseteq\mathcal C^n_\textnormal{qprocess}$.
}
	\label{fig:pt}
\end{figure}
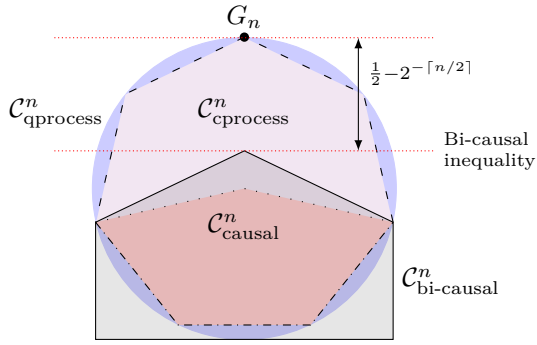
Moreover, we provide evidence that ``acausality'' in these frameworks and Bell non-locality are intimately connected.
The derived game resembles the Ardehali-Svetlichny game~\cite{Svetlichny1987,Ardehali1992}, and shares qualitative and quantitative features.
We show that this analogy also holds in another setting~\cite{Clauser1969,simplestcausalinequalities}.
Our results therefore motivate to further investigate the connection between ``non-causality'' and Bell non-locality, to envisage field-theoretic frameworks with no causal order, and to search for principles with which ``acausality'' is banned in these worlds, in general relativity, and in theories of quantum gravity.

We present our findings in the three subsequent sections.
In the first, we state our results, and in the \mbox{second,} we give a general discussion including the connections to the Ardehali-Svetlichny game and to causal structures.
In the last section, finally, we prove the theorems.

\section{Results}
\label{sec:results}
Before we present our results, we briefly comment on the notation used.
Usually we use calligraphic letters for sets, and uppercase letters for random variables.
The set~$\mathbb Z_n$ is defined as~$\{0,1,\dots,n-1\}$.
If a symbol appears with and without subscripts from some~$\mathbb Z_n$, then the bare symbol denotes the collection under the natural composition, {\it e.g.,}~$A=(A_0,\dots,A_{n-1})$.
We extend this notation further:
If~$\mathcal S\subseteq\mathbb Z_n$, then~\mbox{$A_\mathcal S=(A_k)_{k\in\mathcal S}$} and~\mbox{$A_{\setminus \mathcal S}=(A_k)_{k\in\mathbb Z_n\setminus \mathcal S}$},
and
if~$k\in\mathbb Z_n$, then~\mbox{$A_{\setminus k}=A_{\setminus\{k\}}$}.
We use~$\oplus$ for the addition modulo two, and~$\equiv_2$ for the equivalence relation modulo two.
Finally, if~$\mathcal H$ is a Hilbert space, then~$\mathcal L(\mathcal H)$ denotes the set of linear operators on~$\mathcal H$, and~$\id_\mathcal H$ is the identity operator on~$\mathcal H$.

\subsection{Bi-causal inequalities for arbitrary many parties}
Consider~$n$ parties (regions) where each party~$k\in\mathbb Z_n$ is given a random variable~$X_k$ and outputs a random variable~$A_k$.
Such a setup is described by a behavior (a conditional probability distribution)~$P_{A|X}$.
The possible behaviors~$P_{A|X}$ depend on how the parties interact.
Under the assumption of causal order, {\it i.e.,}~every party can influence her or his future only, one obtains restrictions that are mirrored by causal inequalities~\cite{Abbott2016}.
If some~$P_{A|X}$ violates such an inequality, then~$P_{A|X}$ is not compatible with a causal ordering of the parties and is called {\em non-causal.}
The set of causal correlations among~$n$ parties is~$\mathcal C^n_\textnormal{causal}$.

In a multi-party setting, however, a violation of a~causal inequality could also arise because only some parties violate causal order but not all.
In analogy to multi-party non-local correlations~\cite{Svetlichny1987,Collins2002,Seevinck2002,Gallego2012}, \mbox{Abbott~{\it et al.\/}~\cite{Abbott2017}} show that \mbox{correlations} among~$n$ parties are {\em genuinely multi-party non-causal\/}~(they are non-causal among all parties) if and only if they are not bi-causal:
The correlations cannot be simulated by partitioning the~$n$ parties in two subsets such that the subsets are causally ordered.
\begin{definition}[Bi-causal correlations~\cite{Abbott2017}]
	\label{def:bicausalcorrelations}
	An~$n$-party behavior~$P_{A\mid X}$, for~$n\geq 2$, is {\em bi-causal\/} if and only if
	\begin{align}
		P_{A\mid X} = \!\!\sum_{\emptyset \subsetneq \mathcal K \subsetneq \mathbb Z_n}
		P_K(\mathcal K)
		P_{A_\mathcal K \mid X_\mathcal K, K=\mathcal K}
		P_{A_{\setminus \mathcal K} \mid A_{\mathcal K}, X, K=\mathcal K}
		\,,
	\end{align}
	where~$K$ is a random variable with sample space~\mbox{$\{\mathcal K\mid \emptyset \subsetneq \mathcal K \subsetneq \mathbb Z_n \}$}.
	The set of bi-causal behaviors among~$n$ parties is~$\mathcal C^n_\textnormal{bi-causal}$.
	Behaviors among~$n$~parties that lie outside this set are called {\em genuinely multi-party non-causal.}
\end{definition}
Again, the restrictions imposed by bi-causality are mirrored in bi-causal inequalities.
A behavior~$P_{A|X}$ that violates a bi-causal inequality cannot be decomposed as above, and therefore is {\em genuinely multi-party non-causal.}

We describe a game that is played among~$n$ parties, for arbitrary~$n$, and derive bi-causal inequalities by upper bounding the winning probability for any bi-causal strategy (see Figure~\ref{fig:pt}).
\begin{game*}[$G_n$]
	\label{game}
	Every party~$k\in\mathbb Z_n$ receives a uniformly distributed binary random variable~$X_k$, and must deterministically produce a random variable~$A_k$ that equals~$\omega_k^n(X)$, where
	\begin{align}
		\begin{split}
			\omega_k^n: \mathbb Z_2^n &\rightarrow \mathbb Z_2\\
			x&\mapsto
			\bigoplus_{\substack{i,j\in\mathbb Z_n\setminus\{k\}\\i < j}}
			x_ix_j
			\oplus
			\bigoplus_{i\in\mathbb Z_n}
			\gamma_{k,i} x_i
		\end{split}
	\end{align}
	with
	\begin{align}
		\gamma_{k,i} :=
		\begin{cases}
			1 & (i < k \wedge i\not\equiv_2 k ) \vee (k < i \wedge i\equiv_2 k)\\
			0 & \text{otherwise.}
		\end{cases}
		\label{eq:chi}
	\end{align}
\end{game*}
In the three-party case,~$G_3$ is won whenever
\begin{align}
	\begin{split}
		A_0 = &\left( \neg X_1 \wedge X_2 \right)
		\quad\wedge\quad
		A_1 = \left( \neg X_2 \wedge X_0 \right)
		\quad\wedge
		\\
		&
		\quad
		A_2 = \left( \neg X_0 \wedge X_1 \right)
		\,,
		\label{eq:threepartygame}
	\end{split}
\end{align}
which is the three-party game by Ara{\'{u}}jo and Feix~\cite{AF}, and first published in Ref.~\cite{Baumeler2016}.
Thus, the game~$G_n$ is a generalization of that three-party game to any number of parties.
Another generalization of that game is known~\cite{Araujo2017}, however, with the drawback that the winning probability with causal strategies approaches one.

\begin{theorem}[Bi-causal inequalities]
	\label{thm:bicausalinequality}
	For~$n\geq 2$, the probability of winning the game~$G_n$ with bi-causal correlations~$P_{A\mid X}\in\mathcal C^n_\textnormal{bi-causal}$ is bounded as follows:
	\begin{align}
		\Pr\left[ A=\omega^n(X) \right]
		\leq \frac{1}{2} 
		+ \frac{1}{2^{\ceil{n/2}}}
		\,.
	\end{align}
\end{theorem}
We prove this theorem in Section~\ref{proofs:inequality}.
From Definition~\ref{def:bicausalcorrelations}, moreover, it is evident that~$G_n$ is asymptotically the {\em hardest\/} bi-causal game where at least one party guesses a binary variable.
\begin{lemma}[Least bi-causal winning probability]
	If~$\mathbb Z_n$ is a set of parties,~$d_\textnormal{min}:=\min_{k\in\mathbb Z_n}|\mathcal A_k|$, and ~$\sigma$ a function~$\bigtimes_{k\in\mathbb Z_n}\mathcal X_k\rightarrow \bigtimes_{k\in\mathbb Z_n} \mathcal A_k$, then
	\begin{align}
		\max_{P_{A\mid X}\in\mathcal C^n_\textnormal{bi-causal}}
		\Pr\left[ 
			A = \sigma(X)
		\right]
		\geq
		\frac{1}{d_\textnormal{min}}
		\,.
	\end{align}
\end{lemma}
\begin{proof}
	This lower bound is bi-causally achieved by placing a single party~$k$ with~$d_\textnormal{min}=|\mathcal A_k|$ in the first subset, {\it i.e.,}~$P_K(\{k\})=1$,
	and by letting that party make a uniformly random guess, {\it i.e.,}~$P_{A_k\mid X_k}(a,x)=|\mathcal A_k|^{-1}$ for all~$(a,x)\in\mathcal A_k\times\mathcal X_k$.
	Every other party~$\ell\in\mathbb Z_n\setminus \{k\}$ has access to all random variables~$X_0,\dots,X_{n-1}$ and deterministically generates~$\sigma_\ell(X)$.
\end{proof}

We also establish that the above bound on the probability of winning~$G_n$ with bi-causal behaviors is tight~(see Figure~\ref{fig:pt}).
\begin{theorem}[Faces]
	\label{thm:faces}
	The inequalities of Theorem~\ref{thm:bicausalinequality} represent faces of the bi-causal polytopes:
	The value of the game~$G_n$ with bi-causal behaviors is
	\begin{align}
		\max_{P_{A\mid X}\in\mathcal C^n_{\textnormal{bi-causal}}} \Pr\left[ A=\omega^n(X) \right]
		=
		\frac{1}{2} 
		+ \frac{1}{2^{\ceil{n/2}}}
	\end{align}
	for at least~$n/2$ bi-causal extremal points if~$n$ is even, and at least~$n$ bi-causal extremal points if~$n$ is odd.
\end{theorem}
We prove this theorem in Section~\ref{sec:saturation}.

\subsection{Classical violations}
In the {\em classical-deterministic\/} limit~\cite{Baumeler2016fp} of the process-matrix framework~\cite{Oreshkov2012},
each party obtains a system from the environment, locally applies an arbitrary function~(intervention), and outputs a system to the environment~({\it cf.\/}~Figure~\ref{fig:cc}).
Let~$\mathcal I_k$ be the set of possible states party~$k\in\mathbb Z_n$ can receive from the environment, and~$\mathcal O_k$ the set of possible states party~$k$ can release to the environment.
It is known that the most general dynamics in this setup (without assuming causal order) is described with {\em process functions.}
\begin{definition}[Process function~\cite{Baumeler2016fp,Baumeler2017}]
	An~$n$-party {\em process function\/} is a function~$\omega:\mathcal O \rightarrow \mathcal I$ for some sets~$\mathcal O=\bigtimes_{k\in\mathbb Z_n}\mathcal O_k$,~$\mathcal I=\bigtimes_{k\in\mathbb Z_n}\mathcal I_k$ such that
	\begin{align}
		\forall f \;\exists !i: i=\omega(f(i))
		\,,
	\end{align}
	where~$f=(f_k:\mathcal I_k \rightarrow \mathcal O_k)_{k\in\mathbb Z_n}$ is a collection of functions.
\end{definition}
In words, a process function accounts for the interaction among the parties and has a {\em unique\/} fixed point for each intervention~$f$ of the parties.
Thus, given a choice of interventions~$f$, the process function {\em uniquely\/} determines the states the parties receive from the environment~(see Figure~\ref{fig:cprocess}).
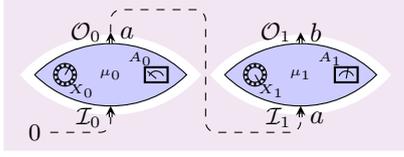
\begin{figure}
	\centering
	\begin{tikzpicture}
		\def\knobx{-0.6}
		\def\knoby{0}
		\def\alpha{45}
		\def\kA{2}
		\def\kB{10}
		\def\mA{140}
		\def\mB{80}
		\coordinate (A) at (0,0);
		\party{A}{\alpha}
		\node at (A) {\tiny $\mu_0$};
		\coordinate (AK) at ($ (A) + (\knobx,\knoby) $);
		\knob{AK}{\kA}{$X_0$}
		\coordinate (AM) at ($ (A) + (-\knobx,-\knoby) $);
		\meter{AM}{\mA}{$A_0$}
		\coordinate (B) at ($ (AR) + (1.5cm, 0) $);
		\party{B}{\alpha}
		\node at (B) {\tiny $\mu_1$};
		\coordinate (BK) at ($ (B) + (\knobx,\knoby) $);
		\knob{BK}{\kB}{$X_1$}
		\coordinate (BM) at ($ (B) + (-\knobx,-\knoby) $);
		\meter{BM}{\mB}{$A_1$}
		\def\offset{0.2}
		\def\offangle{5}
		\def\extrax{0.4}
		\def\extray{1}
		\def\c{violet!10!white}
		\path[draw,\c,fill=\c] ($ (AL) + (-\extrax,\extray) $) -- ($ (AL) + (-\extrax,0) $)
		-- ($ (AL) - (\offset,0) $) to[out=\alpha+\offangle,in=180-\alpha-\offangle] ($ (AR) + (\offset,0) $)
		-- ($ (BL) - (\offset,0) $) to[out=\alpha+\offangle,in=180-\alpha-\offangle] ($ (BR) + (\offset,0) $)
		-- ($ (BR) + (\extrax,0) $)
		-- ($ (BR) + (\extrax,\extray) $)
		-- cycle;
		\path[draw,\c,fill=\c] ($ (AL) + (-\extrax,-\extray) $) -- ($ (AL) + (-\extrax,0) $)
		-- ($ (AL) - (\offset,0) $) to[out=-\alpha-\offangle,in=180+\alpha+\offangle] ($ (AR) + (\offset,0) $)
		-- ($ (BL) - (\offset,0) $) to[out=-\alpha-\offangle,in=180+\alpha+\offangle] ($ (BR) + (\offset,0) $)
		-- ($ (BR) + (\extrax,0) $) -- ($ (BR) + (\extrax,-\extray) $)
		-- cycle;
		\draw[>=stealth,double,->] (AT.center) -- ++(0,0.15) node[above,left] {\small $\mathcal O_0$} node[above,right] {$a$};
		\draw[>=stealth,double,<-] (AB.center) -- ++(0,-0.15) node[below,left] {\small $\mathcal I_0$};
		\draw[>=stealth,double,->] (BT.center) -- ++(0,0.15) node[above,left] {\small $\mathcal O_1$} node[above,right] {$b$};
		\draw[>=stealth,double,<-] (BB.center) -- ++(0,-0.15) node[below,left] {\small $\mathcal I_1$} node[below,right] {$a$};
		\path (A) -- (B) coordinate [midway] (M);
		\draw[rounded corners,dashed] ($ (AT) + (0,0.15) $)
		-- ++(0,0.3)
		-| (M)
		|- ($ (BB) - (0,0.35) $)
		-- ++(0,0.3);
		\node at ($ (AB) - (1,0.35) $) (z) {\small$0$};
		\draw[rounded corners,dashed] (z) -| ($ (AB) - (0,0.15) $);
	\end{tikzpicture}
	\caption{
		The process function is a function from~$\mathcal O$ to~$\mathcal I$.
		The dashed connections illustrates an example of a two-party process function where the left party is in the causal past of the right party, {\it i.e.,}~$\forall (a,b)\in\mathcal O_0\times\mathcal O_1: (a,b) \mapsto(0,a)$.
	}
	\label{fig:cprocess}
\end{figure}
This is intuitive: No fixed point corresponds to a logical contradiction, and multiple fixed points to an ambiguity~\cite{Baumeler2016fp}.\footnote{In fact, the above definition is tantamount to the requirement that for all~$f$ there exists {\em at least one\/} fixed point (no contradiction)~\cite{Baumeler2017}, and tantamount to the requirement that for all~$f$ there exists {\em at most one\/} fixed point (no ambiguity)~\cite{Baumeler2020}.}
For three parties or more, there exist process functions that do not reflect a causal order among the parties (Equation~\eqref{eq:threepartygame} interpreted as a function~$(X_0,X_1,X_2)\mapsto(A_0,A_1,A_2)$ is an example).

The parties in this classical-deterministic world have access to a process function to generate some behavior~$P_{A\mid X}$.
Each party~$k$ receives an element from the set~$\mathcal I_k$ and some~$X_k$, and applies a function (intervention)~$\mu_k:\mathcal X_k\times\mathcal I_k\rightarrow\mathcal A_k\times\mathcal O_k$ to generate the output~$A_k$ and the system that is released to the environment.
The process function illustrated in Figure~\ref{fig:cprocess}, for instance, leads to behaviors~$P_{A\mid X}$ where party~$0$ is in the causal past of party~$1$.
\begin{definition}[Process-function behavior]
	An~$n$-party behavior~$P_{A\mid X}$ is a {\em deterministic process-function behavior\/} if and only if there exists some~$n$-party process function~$\omega$ and interventions~$\mu=(\mu_0,\dots,\mu_{n-1})$ such that
	\begin{align}
		P_{A\mid X}(a,x) =
		\begin{cases}
			1 & \exists i: \alpha(x,i) = a \wedge \omega(\beta(x,i)) = i\\
			0 & \text{otherwise,}
		\end{cases}
	\end{align}
	where~$\alpha_k: \mathcal X_k\times\mathcal I_k\rightarrow\mathcal A_k$, and~$\beta_k:\mathcal X_k\times\mathcal I_k\rightarrow\mathcal O_k$ are the components of~$\mu_k=(\alpha_k,\beta_k)$.
	The set of~$n$-party {\em process-function behaviors\/} is the convex hull of all~$n$-party deterministic process-functions behaviours and is denoted by~$\mathcal C^n_{\textnormal{cprocess}}$.
\end{definition}
\begin{theorem}
	\label{thm:lc}
	The function~$\omega^n$ of the game~$G_n$ is an~$n$-party process function.
\end{theorem}
This theorem---proven in Section~\ref{proofs:violation}---immediately implies that the game~$G_n$ is won deterministically in this framework~(see Figure~\ref{fig:pt}): 
\begin{corollary}[Non-causal value of~$G_n$]
	The value of the game~$G_n$ with classical-process behaviors is
	\begin{align}
		\max_{P_{A\mid X}\in\mathcal C^n_{\textnormal{cprocess}}} \Pr\left[ A=\omega^n(X) \right]
		= 1
		\,.
	\end{align}
\end{corollary}
\begin{proof}
	Let the~$n$ parties have access to the process function~$\omega^n$ of the game~$G_n$.
	Each party~\mbox{$k\in\mathbb Z_n$} relays the input~$x_k$ to the environment~$o_k=x_k$, and uses the input~$i_k$ from the environment as guess~$a_k$.
	Formally, party~$k$ implements the function~\mbox{$\mu_k:(x_k,i_k) \mapsto (i_k,x_k)$,}
	by which the behavior~$P_{A\mid X}$ that equals~$\omega^n$ is obtained.
\end{proof}
This corollary shows that for any number of parties, the classical-deterministic limit of the process-matrix framework leads to genuinely multi-party non-causal correlations~($\forall n\geq 3: \mathcal C^n_\textnormal{cprocess}\not\subseteq \mathcal C^n_\textnormal{bi-causal}$),
and that the violation is bounded by a constant ($\forall n\geq 3:$ the gap is at least~$1/4$).
While previous games~\cite{Abbott2017} share the former feature, their non-causal correlations are limited---the gap vanishes for increasing number of parties.

\subsection{Quantum violations}
In the process-matrix framework~\cite{Oreshkov2012}, each party~$k$ receives a quantum state on the Hilbert space~$\mathcal I_k$, applies a quantum instrument~$\mu_k=\{\mu_k^{a,x}\}_{(a,x)\in\mathcal A_k\times \mathcal X_k}$, and releases a quantum state on the Hilbert space~$\mathcal O_k$.
A~quantum instrument~$\mu_k$ is a family of completely positive trace-non-increasing maps from~$\mathcal L(\mathcal I_k)$ to~$\mathcal L(\mathcal O_k)$ such that~$\forall x\in \mathcal X_k$, the map~$\sum_{a\in\mathcal A_k} \mu_k^{a,x}$ is trace preserving.
\begin{definition}[Process matrix and process-matrix behaviors~\cite{Oreshkov2012}]
	An~$n$-party {\em process matrix\/} is a positive-semi-definite matrix~$W\in\mathcal L(\mathcal O\otimes\mathcal I)$ for some Hilbert spaces~$\mathcal O=\bigotimes_{k\in\mathbb Z_n}\mathcal O_k$,~$\mathcal I=\bigotimes_{k\in\mathbb Z_n}\mathcal I_k$ such that
	\begin{align}
		\forall\mu:
		\Tr\left[\bigotimes_{k\in\mathbb Z_n} M_k W\right] = 1
		\,,
		\label{eq:pm}
	\end{align}
	where~\mbox{$\mu=(\mu_k:\mathcal L(\mathcal I_k)\rightarrow\mathcal L(\mathcal O_k))_{k\in\mathbb Z_n}$} is a family of completely positive trance-preserving maps, and where~$M_k$ is the Choi operator\footnote{Note that in the process-matrix framework, the Choi operator of a map~$\mathcal E$ is defined as~$[\id\otimes\mathcal E(\doubleket\id\doublebra\id)]^T$, whereas some define the Choi operator with a {\em partial\/} transpose only, or without transpose. For our results, this distinction is irrelevant.}~\cite{Choi1975,Jamiokowski1972} of~$\mu_k$.
	An~$n$-party behavior~$P_{A\mid X}$ is a {\em process-matrix behavior\/} if and only if there exists some~$n$-party process matrix~$W$ and quantum instruments~$\mu=(\mu_k)_{k\in\mathbb Z_n}$ such that
	\begin{align}
		P(a\mid x) = Tr\left[\bigotimes_{k\in\mathbb Z_n} M_k^{a_k,x_k}W\right]
		\,.
	\end{align}
	The set of~$n$-party process-matrix behaviors is~$\mathcal C^n_\textnormal{qprocess}$.
\end{definition}
Just as in the classical case, the process matrix accounts for the interaction among the parties.
A~process matrix yields a conditional probability distribution under any choice of quantum instruments (interventions) of the parties:
No matter what experiment the parties perform, and even if they share entangled states,\footnote{Allowing
	the parties to share entangled states forces~$W$ to be positive semi-definite.
	Oreshkov, Costa, and Brukner~\cite{Oreshkov2012} show that if the Hilbert space~$\mathcal O$ is trivial and the parties share arbitrary entangled states, then every process matrix is a quantum state (see~Figure~\ref{fig:nonsignaling}).
	In contrast, Barnum {\it et al.\/}~\cite{Barnum2010} and Ac\'in {\it et al.\/}~\cite{Acin2010} show that if~$\mathcal O$ is trivial but the parties {\em do not\/} share quantum states, then~$W$ is positive on pure tensors as opposed to positive semi-definite; for three parties or more, this yields correlations beyond the quantum set.}
the probabilities of their observations are well-defined.
The example illustrated in Figure~\ref{fig:cprocess} is obtained by the process matrix~$W=\ket 0\bra 0_{I_0}\otimes \doubleket{\id}\doublebra{\id}_{O_0,I_1}\otimes\id_{O_1}$, where~$\doubleket{\id}$ is the non-normalized maximally entangled state~$\sum_o\ket o\ket o$, where~$\mathcal O_0\simeq\mathcal I_1$, and the sum is taken over a basis of~$\mathcal O_0$.

\begin{theorem}
	\label{thm:qc}
	If~$\omega:\mathcal O\rightarrow\mathcal I$ is an~$n$-party process function, then
	\begin{align}
		W:=
		\sum_{o\in\mathcal O}
		\ket o
		\bra o_O
		\otimes
		\ket{\omega(o)}
		\bra{\omega(o)}_I
		\label{eq:bijection}
	\end{align}
	is an~$n$-party process matrix.
\end{theorem}
This theorem---shown in Section~\ref{proofs:quantum}---implies that every process-function behavior is a process-matrix behavior~(see Figure~\ref{fig:pt}), and therefore
\begin{align}
	\max_{P_{A\mid X}\in\mathcal C^n_\textnormal{qprocess}} \Pr\left[ A=\omega^n(X) \right]
	= 1
	\,.
\end{align}
Moreover, from Ref.~\cite{Baumeler2016fp} it is known that every process function and every mixture of process functions is embeddable into a {\em reversible\/} process function with two additional parties:
A party in the global past with~a~trivial input, and a party in the global future with a trivial output.
By the above theorem, the same behaviours are unitarily extensible~\cite{Araujo2017purification}---the corresponding process matrices can be extended to unitary dynamics.
This does not hold for all process matrices.
Moreover, Barrett, Lorenz, and Oreshkov~\cite{Barrett2020} show that every unitarily extensible two-party process matrix is causal.

Note that the above theorem is not trivial:
It is not proven by referring to the {\em classical-deterministic limit\/} of process matrices.
The reason for this is the following.
Call an operator~$W$ a {\em deterministic-diagonal process matrix\/} if and only if~$W$ has zero-one entries only and satisfies Equation~\eqref{eq:pm} {\em restricted\/} to interventions~$\bigotimes_{j\in\mathbb Z_n} M_j^{a_j,x_j}$ diagonal in the same basis as~$W$.
The limit theorem~\cite{Baumeler2016,Baumelerphd} states that there exists a bijection (through Equation~\eqref{eq:bijection}) between process functions and deterministic-diagonal process matrices.
{\em A~priori,\/} however, it could be that some deterministic-diagonal process matrices are {\em not\/} process matrices;
it could be that there exist {\em quantum\/} instruments~$\mu$ such that the probabilities are not well-defined for some deterministic-diagonal process matrices.
This potentiality arises because the {\em enlargement\/} of the set of possible interventions from functions to quantum instruments {\em restricts\/} the set of objects~$W$ that satisfy Equation~\eqref{eq:pm}.
While a few process functions and stochastic processes~(probabilistic generalizations of process functions) were translated {\em case-by-case\/} to process matrices~\cite{ClassicalNC1,Baumelerphd,Araujo2017purification}, the above theorem was unknown.

\section{Discussion}
Violations of causal order arise naturally---not in the sense that such violations are known to arise in our physical world, but in the sense that they arise and are advocated by our best physical theories.
This leaves open a binary alternative:
Either causal order does {\em not\/} hold in our physical world, or these violations are mathematical artifacts of our theories.
To show the latter, one must describe a reasonable and physical principle from which causal order is reestablished.
All attempts so far, however, fail to do so, and leave us with the ``chronology protection conjecture''~\cite{Hawking1992}.
While this conjecture is well motivated, it also leaves room to speculate that {\em ``acausal''\/} dynamics arise in yet unprobed physical regimes, {\it e.g.,}~in the minuscule of quantum foam or in black holes~\mbox{\cite{Thorne1991a,Thorne,Svetlichny2005,wheeler2010geons}}.
Overall, our results show that causal order cannot be derived in the classical limit~\cite{ClassicalNC1}, with the requirement of determinism and reversibility~\cite{Baumeler2016fp}, with the~$\cc{NP-hardness}$ assumption~\cite{Aaronson2005,Baumeler2018}, and also not---as we show here---in the many-body limit.
Note that the process functions and process matrices described here can be made reversible~\cite{Baumeler2016,Araujo2017purification}.
In the context of general relativity~\cite{Baumeler2017,Tobar2020}, our findings present non-trivial and reversible closed time-like curves that traverse any number of local space-time regions.

\subsection{Relation to non-locality}
We not only rule out that ``acausal'' processes must be restricted to few regions, but also establish a link between non-causal and non-local correlations.
To see this, we briefly describe the Ardehali-Svetlichny non-local game.\footnote{The three-body Bell inequality was introduced by Svetlichny~\cite{Svetlichny1987} and generalized by Ardehali~\cite{Ardehali1992} based on insights by Mermin~\cite{Mermin1990}. A reformulation of these inequalities as a game is found in the work by Ambainis {\it et al.\/}~\cite{Ambainis2010,Ambainis2013}.}
This game is played among~$n$ parties that are space-like separated---they cannot communicate (see Figure~\ref{fig:nonsignaling}).
\begin{figure}
	\centering
	\begin{tikzpicture}
		\def\knobx{-0.6}
		\def\knoby{0}
		\def\alpha{45}
		\def\kA{2}
		\def\kB{10}
		\def\kC{5}
		\def\mA{140}
		\def\mB{80}
		\def\mC{60}
		\coordinate (A) at (0,0);
		\party{A}{\alpha}
		\node at (A) {\tiny $\mu^{a,x}$};
		\coordinate (AK) at ($ (A) + (\knobx,\knoby) $);
		\knob{AK}{\kA}{$X$}
		\coordinate (AM) at ($ (A) + (-\knobx,-\knoby) $);
		\meter{AM}{\mA}{$A$}
		\coordinate (B) at ($ (AR) + (1.5cm, 0) $);
		\party{B}{\alpha}
		\node at (B) {\tiny $\nu^{b,y}$};
		\coordinate (BK) at ($ (B) + (\knobx,\knoby) $);
		\knob{BK}{\kB}{$Y$}
		\coordinate (BM) at ($ (B) + (-\knobx,-\knoby) $);
		\meter{BM}{\mB}{$B$}
		\coordinate (C) at ($ (BR) + (1.5cm, 0) $);
		\party{C}{\alpha}
		\node at (C) {\tiny $\tau^{c,z}$};
		\coordinate (CK) at ($ (C) + (\knobx,\knoby) $);
		\knob{CK}{\kC}{$Z$}
		\coordinate (CM) at ($ (C) + (-\knobx,-\knoby) $);
		\meter{CM}{\mC}{$C$}
		\def\offset{0.2}
		\def\offangle{5}
		\def\extrax{0.4}
		\def\extray{1}
		\def\c{violet!10!white}
		\path[draw,\c,fill=\c] ($ (AL) + (-\extrax,-\extray) $) -- ($ (AL) + (-\extrax,0) $)
		-- ($ (AL) - (\offset,0) $) to[out=-\alpha-\offangle,in=180+\alpha+\offangle] ($ (AR) + (\offset,0) $)
		-- ($ (BL) - (\offset,0) $) to[out=-\alpha-\offangle,in=180+\alpha+\offangle] ($ (BR) + (\offset,0) $)
		-- ($ (CL) - (\offset,0) $) to[out=-\alpha-\offangle,in=180+\alpha+\offangle] ($ (CR) + (\offset,0) $)
		-- ($ (CR) + (\extrax,0) $) -- ($ (CR) + (\extrax,-\extray) $)
		-- cycle;
		\draw[>=stealth,double,<-] (AB.center) -- ++(0,-0.15) node[below,left] {\small $\mathcal I_A$};
		\draw[>=stealth,double,<-] (BB.center) -- ++(0,-0.15) node[below,left] {\small $\mathcal I_B$};
		\draw[>=stealth,double,<-] (CB.center) -- ++(0,-0.15) node[below,left] {\small $\mathcal I_C$};
	\end{tikzpicture}
	\caption{
		Each party (region) performs an experiment on a~system.
		The parties, however, are assumed to be space-like separated, and therefore cannot communicate.
		In a classical setup, the systems are described by random variables, and in the quantum case, by quantum states.
	}
	\label{fig:nonsignaling}
\end{figure}
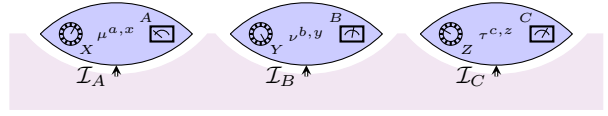
\begin{definition}[Ardehali-Svetlichny function]
	\label{def:svetfunction}
	An~{\em $m$-ary Ardehali-Svetlichny function\/} is a function~$S_\lambda^m$ where
	\begin{align}
		\begin{split}
			S^m_\lambda: \mathbb Z_2^m &\rightarrow \mathbb Z_2\\
			z &\mapsto
			\bigoplus_{\substack{i,j\in\mathbb Z_m\\i<j}}
			z_iz_j
			\oplus
			\bigoplus_{i\in\mathbb Z_m}
			\lambda_i
			z_i
			\,,
		\end{split}
	\end{align}
	for some~$\lambda\in\mathbb Z_2^m$.
\end{definition}
\begin{asgame*}
	Every party~$k\in\mathbb Z_n$ receives a uniformly distributed binary random variable~$X_k$ and must produce a random variable~$A_k$ such that
	\begin{align}
		\bigoplus_{i\in\mathbb Z_n} A_i = S_\lambda^n(X)
		\,.
	\end{align}
\end{asgame*}
\begin{lemma}[Ardehali-Svetlichny inequalities~\cite{Collins2002,Seevinck2002,Ardehali1992}]
	\label{lemma:svetlichny}
	The winning probability
	\begin{align}
		\Pr\left[ \bigoplus_{i\in\mathbb Z_n} A_i = S_\lambda^n(X) \right]
	\end{align}
	of the Ardehali-Svetlichny game for~$n\geq 2$ is
	\begin{itemize}
		\item upper bounded by~$1/2 + 1/2^{\floor{n/2}+1}$ for local behaviors~$P_{A\mid X}$,
		\item upper bounded by~$3/4$ for Svetlichny bi-local behaviors~$P_{A\mid X}$,
		\item and reaches the Tsirelon bound~\cite{TsirelsonsBound}~$(2+\sqrt 2)/4$ for quantum behaviors~$P_{A\mid X}$.
	\end{itemize}
\end{lemma}

A {\em first\/} connection between the game~$G_n$ and the Ardehali-Svetlichny game is that~$\omega^n_k(x) =  S^{n-1}_\lambda(x_{\setminus k})$ for~$\lambda$ being the alternating string~$(0,1,0,1,\dots)$ when~$k$ is even, and~$(1,0,1,0,\dots)$ when~$k$ is odd.
While in~$G_n$, {\em each\/} party must guess the value of an Ardehali-Svetlichny function, in the Ardehali-Svetlichny game, the parties must {\em jointly\/} guess an Ardehali-Svetlichny function.
A {\em second\/} connection is that~$G_n$ can be used to certify genuinely multi-party {\em non-causal\/} correlations, and the Ardehali-Svetlichny game can be used to certify genuinely multi-party {\em non-local\/} correlations.\footnote{Note that there are alternative definitions of genuinely multi-party non-local correlations~\cite{Bancal2013,Gallego2012}.
The definition as {\em not Svetlichny bi-local\/} is, however, the most conservative one.}
{\em Third,\/} the bi-causal (and therefore also the causal) bound of~$G_n$ is essentially the local bound of the Ardehali-Svetlichny game.
And {\em finally,\/} with the frameworks employed, the parties win~$G_n$ with constant probability, as it is the case for the Ardehali-Svetlichny game in quantum theory.
This last feature implies that non-local and non-causal correlations are {\em unlimited\/} and yield an exponential non-local versus local, and respectively non-causal versus causal advantage for an increasing number of parties~$n$.

Note that we cannot simply connect the results on \mbox{(non-)locality} and (a)causality by considering the {\em same\/} game, but must aim for a {\em translation.\/} 
A reason for this is that the infamous game where every party must guess its neighbours input (GYNI) does not allow for a quantum-over-classical advantage with shared resources~\cite{Almeida2010}, however, it allows for a non-causal-over-causal advantage when played with process matrices~\cite{simplestcausalinequalities}.
The similarities displayed above, however, lead us to ask whether any parity local game is {\em translatable\/} to a causal game.
A parity local game is a game where the parity of the parties' outputs must equal a function of their inputs.
A~proposal for general translations to causal games is to ask {\em each\/} party separately to guess the function value.
We briefly apply this recipe to the CHSH game~\cite{Clauser1969}.
In that two-party parity game, each party is given a binary random variable~$X$ and~$Y$, and produces~$A$~and~$B$, respectively, satisfying~\mbox{$A\oplus B=XY$}.
By following the above recipe, we ask the parties to produce~$A$~and~$B$ such that~$A=XY$ and~$B=XY$.
This resulting game is the {\em lazy guess-your-neighbour's-input game\/}~\cite{simplestcausalinequalities}:~$X(A\oplus Y)=Y(B\oplus X)=0$.
The causal bound for this game is~$3/4$ and coincides with the local bound of the CHSH game we started with.
Two parties using the process-matrix framework violate that causal bound and win this derived game with probability~$0.819$ (the maximal winning probability is unknown)~\cite{simplestcausalinequalities}, the CHSH game is quantum mechanically won with probability at most~$(2+\sqrt 2)/4\approx 0.854$~\cite{TsirelsonsBound}.

\subsection{Geometric and causal structure}
By Theorem~\ref{thm:faces}, the bi-causal inequalities describe faces of the bi-causal polytopes.
However, these inequalities {\em do not\/} describe facets.
This can be seen for~$n=3$.
The hyperplane specified by that inequality is~$7$-dimensional:
\begin{align}
	\begin{split}
		&p_{000}^{000} +
		p_{010}^{001} +
		p_{001}^{010} +
		p_{001}^{011} + \\
		&p_{010}^{100} +
		p_{100}^{101} +
		p_{010}^{110} +
		p_{000}^{111} = 6
		\,,
	\end{split}
\end{align}
where~$p_{abc}^{xyz}:=P_{A, X}(a,b,c,x,y,z)$.
However, there are exactly three bi-causal extremal points that lie on that hyperplane.
For~$n$ parties, the hyperplane specified by the bi-causal inequality is~$(2^n-1)$-dimensional.
We have identified~$n/2$ for~$n$ even, and~$n$ for~$n$ odd, bi-causal extremal points on that hyperplane (see Theorem~\ref{thm:faces} and the proof in Section~\ref{sec:saturation} for the strategies), and leave open whether more optimal strategies exist for~$n\geq 4$

Following the work by Barrett, Lorenz, and Oreshkov~\cite{Barrett2020}, we are in position to describe the causal structure corresponding to the process function~$\omega^n$.
In their article, the authors show that the causal structure of~$\omega^3$ is the fully connected directed graph as shown in Figure~\ref{fig:graph3}, where each node denotes a party, and where an edge~$i\rightarrow j$ indicates that party~$i$ influences party~$j$.
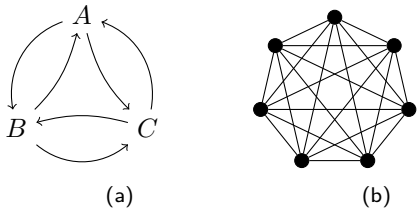
\begin{figure}
	\centering
	\begin{subfigure}[t]{.4\columnwidth}
		\begin{tikzpicture}
			\def\of{40}
			\def\ofi{15}
			\node at (90:1cm) (A) {$A$};
			\node at (90+120:1cm) (B) {$B$};
			\node at (90+240:1cm) (C) {$C$};
			\draw[->] (A) to[out=-120-\of,in=60+\of] (B);
			\draw[->] (B) to[out=0-\of,in=180+\of] (C);
			\draw[->] (C) to[out=120-\of,in=-60+\of] (A);
			\draw[<-] (A) to[out=-120+\ofi,in=60-\ofi] (B);
			\draw[<-] (B) to[out=0+\ofi,in=180-\ofi] (C);
			\draw[<-] (C) to[out=120+\ofi,in=-60-\ofi] (A);
		\end{tikzpicture}
		\caption{\label{fig:graph3}}
	\end{subfigure}
	\begin{subfigure}[t]{.4\columnwidth}
		\begin{tikzpicture}
			\graph[nodes={draw, circle,fill,inner sep=0pt,minimum size=3pt}, clique, n=7, clockwise, radius=1cm]
			{
				1/"$\cdot$",
				2/"$\cdot$",
				3/"$\cdot$",
				4/"$\cdot$",
				5/"$\cdot$",
				6/"$\cdot$",
				7/"$\cdot$"
			};
		\end{tikzpicture}
		\caption{\label{fig:graphK}}
	\end{subfigure}
	\caption{(a) The causal structure of~$\omega^3$ is the fully connected directed graph~\cite{Barrett2020}.
		(b) If undirected edges represent bi-directional edges, then the causal structure of~$\omega^n$ is given by the complete graph~$K_n$ (here,~$n=7$).}
\end{figure}
For~$n$ parties, the causal structure again is such a graph.
If every undirected edge represents a bi-directional edge, then the causal structure of~$\omega^n$ is given by the complete graph~$K_n$ (see Figure~\ref{fig:graphK}).
Every party~$i\in\mathbb Z_n$ influences {\em every other\/} party~$j\in\mathbb Z_n\setminus \{i\}$,
and, no matter what intervention~$f=(f_0,\dots,f_n)$ the parties perform, no party~$k$ influences her or himself.

\subsection{Open questions}
Our findings bring forward a series of open questions.
While the main question {\em ``how to derive causal order without presupposing causal order''\/} stated in the introduction persists, we now ask whether a general translation between Bell non-local correlations and non-causal correlations is possible.
Such a translation enriches the understanding of the ``non-causal'' world---just as the translation between Bell non-locality and contextuality~\cite{Stairs1983,Heywood1983,Renner,Renner2012,Brassard2005a,Cabello2020}~(see Budroni~{\it et al.\/}~\cite{Budroni2021} for a recent preprint on that topic)---and speculatively brings forward experimental setups to violate causal order~({\it cf.\/}~Oreshkov~\cite{Oreshkov2019}, Purves and Short~\cite{Purves2021}, Wechs {\it et al.\/}~\cite{Wechs2021}, and Wechs, Branciard, and Oreshkov~\cite{Wechs2022}).

Because the process functions presented here can be embedded in reversible functions~\cite{Baumeler2016fp}, they are physically implementable {\it e.g.,}~by the means of billiard-ball collisions~\cite{Fredkin1982}.
This yields a physical approach to produce the presented correlations:
Construct closed time-like curves, {\it e.g.,}~with wormholes~\cite{Morris1988,Novikov1989,Frolov1990}, which connect the output of the process function~$\omega^n$ with the past space-time boundaries of the~$n$ regions.
The general-relativistic properties of these space-time geometries, however, are unknown.
More details on this approach are found in the articles by \mbox{Baumeler~{\it et al.\/}~\cite{Baumeler2017},} and by Tobar and Costa~\cite{Tobar2020}.
This latter article moreover provides process functions inequivalent to those presented here for~$n=4$ regions.

To bridge a gap still present between the process-matrix framework and relativity, and towards a theory of quantum gravity, one can ask whether a {\em field-theoretic\/} version of the process-matrix framework allows for violations of causal order.
In such a framework, every localized space-time region can be regarded as a party, and thus, parties may be {\em overlapping.}
The derived non-causal correlations for any number of parties support such a possibility.
In connection to that, Giacomini, Castro-Ruiz, and Brukner~\cite{Giacomini2016} describe process matrices with {\em continuous\/}-variable systems.

Moreover, while we show that the causal structure of the process function~$\omega^n$ is the complete graph~$K_n$, we leave open the question of singling out the causal structures that are attainable with process functions:
Does there exist a graph-theoretic criterion to reject a function as a process function?
The absence of self-loops is a necessary condition (see~Lemma~\ref{lemma:constant} below) but insufficient.

Finally, it is intriguing to further compare the process-matrix framework with process functions.
The reason is that the only known {\em deterministic\/} violations of causal inequalities are present in the quantum and in the classical framework (this is also the case here).
Moreover, all unitarily extendible two-party process matrices are causal~\cite{Barrett2020}, just as it is in the classical case.
Thus, we ask: Is~$\mathcal C^n_\textnormal{qprocess}\setminus \mathcal C^n_\textnormal{causal} \cap\mathcal D = \mathcal C^n_\textnormal{cprocess}\setminus \mathcal C^n_\textnormal{causal} \cap\mathcal D$, where~$\mathcal D$ is the set of deterministic behaviors?
An affirmative answer would imply that in the ``acausal'' regime, quantum and classical theories are equivalent~({\it cf.\/}~Aaronson and Watrous~\cite{Aaronson2009}).

\section{Proofs}
\subsection{Derivation of bi-causal inequalities}
\label{proofs:inequality}
We derive the bi-causal inequalities (Theorem~\ref{thm:bicausalinequality}) in two main steps.
Note that in any bilateral partition~\mbox{$(\mathcal K,\mathbb Z_n\setminus \mathcal K)$} of~$n$ parties, where the parties in~$\mathcal K$ causally precede the remaining parties, every party~$j$ not in~$\mathcal K$ can deterministically guess the random variable~$\omega^n_j(X)$; only the parties in~$\mathcal K$ have to make non-trivial guesses.
Since the parties in~$\mathcal K$ can communicate as they wish, it might be the case that a single party only has to make a non-trivial guess, and that all other parties can base their guesses on that party.
First, we show that this almost never happens.
In the second step, we invoke the bound on the winning probability of the Ardehali-Svetlichny game.

\subsubsection{Guesses cannot be recycled}
We define the event~$\mathcal E^\mathcal K$:
There is some party~$k_0\in\mathcal K$ such that every party~$\ell\in\mathcal K$ can deterministically compute~$\omega_\ell^n(X)$ from~$\omega_{k_0}^n(X)$ and~$X_\mathcal K$.
In other words, there exists some~$k_0\in\mathcal K$ such that the term~$\omega_{k_0}^n(x)\oplus\omega_\ell^n(x)$ for any~$\ell\in\mathcal K$ is {\em independent\/} of the variables~$x_{\setminus \mathcal K}$.
Note that if such a~$k_0$ exists, then clearly~$\forall k,\ell\in\mathcal K$ the expression~$\omega_k^n(x)\oplus\omega_\ell^n(x)$ is independent of~$x_{\setminus \mathcal K}$.
\begin{definition}
	For a non-empty~$\mathcal K\subsetneq\mathbb Z_n$ and~$n\geq 2$, the event~$\mathcal E^\mathcal K$ is
	\begin{align}
		\begin{split}
			\mathcal E^\mathcal K &:=
			\bigg\{
				x_\mathcal K\in\mathbb Z_2^{|\mathcal K|}
				\mid
				\forall k,\ell\in\mathcal K\,
				\exists c\in\mathbb Z_2\,
				\\
				&
				\forall x_{\setminus \mathcal K}\in\mathbb Z_2^{n-|\mathcal K|}
				:
				\omega_k^n(x)\oplus\omega_\ell^n(x)
				=
				c
			\bigg\}
			\,.
		\end{split}
	\end{align}
\end{definition}
We give an upper bound on the probability for this event to occur.
\begin{lemma}
	\label{lemma:occurenceE}
	Let~$\mathcal K\subsetneq\mathbb Z_n$ be non-empty,~$n\geq 2$ and let~$X_\mathcal K$ be a uniformly distributed random variable.
	The probability of the event~$\mathcal E^\mathcal K$ over~$\mathcal X_\mathcal K$ is upper bounded as follows:
	\begin{align}
		\Pr_{\mathcal X_\mathcal K}\left[ 
			\mathcal E^\mathcal K
		\right]
		\leq
		2^{-|\mathcal K|+1}
		\,.
	\end{align}
\end{lemma}
\begin{proof}
	We express~$\omega_k^n(x)$ in a form where the terms involving~$x_\mathcal K$ are separated from the rest:\footnote{We make use of the notation~$i\not\in\mathcal K:\Leftrightarrow i\in\mathbb Z_2^n\setminus\mathcal K$ if no danger of confusion exists.}
	\begin{align}
		\omega_k^n(x) =
		\bigoplus_{\substack{
			i,j\not\in\mathcal K\\
			i<j}}
		x_ix_j
		\oplus 
		\bigoplus_{
			i\not\in\mathcal K
		}
		\alpha_{k,i} x_i
		\oplus
		\beta_k
		\,,
		\label{eq:revised}
	\end{align}
	with
	\begin{align}
		\alpha_{k,i} &:= \gamma_{k,i}
		\oplus
		\bigoplus_{j\in\mathcal K\setminus\{k\}}
		x_j
		\,,
		\\
		\beta_k &:=
		\bigoplus_{\substack{
			i,j\in\mathcal K\setminus\{k\}\\
			i<j
		}}
		x_ix_j
		\oplus
		\bigoplus_{i\in\mathcal K}
		\gamma_{k,i} x_i
		\,,
	\end{align}
	where~$\gamma_{k,i}$ is defined as in Equation~\eqref{eq:chi}.
	This allows us to express~$\omega_k^n(x)\oplus\omega_\ell^n(x)$, for~$k,\ell\in\mathcal K$, compactly as
	\begin{align}
		\bigoplus_{i\not\in\mathcal K}
		x_i\left( 
		\gamma_{k,i} \oplus \gamma_{\ell,i} \oplus x_k \oplus x_\ell
		\right)
		\oplus
		\beta_k
		\oplus
		\beta_\ell
		\,.
	\end{align}
	Therefore, if~$x_\mathcal K\in\mathcal E^\mathcal K$, then
	\begin{align}
		\forall k,\ell\in\mathcal K,i\not\in\mathcal K:
		\gamma_{k,i} \oplus \gamma_{\ell,i} = x_k \oplus x_\ell
		\,,
	\end{align}
	and moreover
	\begin{align}
		\forall k,\ell\in\mathcal K, i,j\not\in\mathcal K: \gamma_{k,i} \oplus \gamma_{\ell,i} = \gamma_{k,j} \oplus \gamma_{\ell,j}
		\,.
	\end{align}
	Thus, for some~$x_\mathcal K \in\mathcal E^\mathcal K$ we can define~$c_{k,\ell}:=\gamma_{k,i_0} \oplus \gamma_{\ell,i_0}$ for an arbitrary~$i_0\not\in\mathcal K$.
	Since~$c_{k,\ell}$ is {\em independent\/} of~$x_\mathcal K$, we have
	\begin{align}
		\forall x_\mathcal K:
		x_\mathcal K\in\mathcal E^\mathcal K \Longrightarrow x_k\oplus x_\ell = c_{k,\ell}
		\,;
	\end{align}
	the bits of every sequence~$x_\mathcal K\in\mathcal E^\mathcal K$ are related by the {\em same\/} constants~$c_{k,\ell}$.
	Now, we pick some~$k_0\in\mathcal K$ and find that~$x_\ell$ for every~$\ell\in\mathcal K$ is uniquely determined by~$x_{k_0}$;
	there are at most two distinct sequences~$x_\mathcal K$ in~$\mathcal E^\mathcal K$.
	Finally, since~$X_\mathcal K$ is uniformly distributed, we obtain
	\begin{align}
		\Pr_{X_\mathcal K}\left[ \mathcal E^\mathcal K \right]
		=
		\frac{|\mathcal E^\mathcal K|}{2^{|\mathcal K|}}
		\leq
		2^{-|\mathcal K|+1}
		\,.
	\end{align}
\end{proof}

\subsubsection{Proof of Theorem~\ref{thm:bicausalinequality} via an application of the Ardehali-Svetlichny inequality}
We prove the bi-causal inequalities by conditioning the winning probabilities on the event~$\mathcal E^\mathcal K$.
In the event~$\mathcal E^\mathcal K$, only one party must guess an Ardehali-Svetlichny function, and otherwise,
{\em at least two\/} parties must guess independent Ardehali-Svetlichny functions.
It turns out that these independent functions are {\em parity relative.}
\begin{lemma}[Guessing two parity-relative functions]
	\label{lemma:parity-relative}
	If~$f:\mathbb Z_2^m\rightarrow\mathbb Z_2$ and~$g:\mathbb Z_2^m\rightarrow\mathbb Z_2$ are two functions that satisfy
	\begin{align}
		g(z) = f(z) 
		\oplus 
		c
		\oplus 
		\bigoplus_{i\in\mathcal S} z_i
		\,,
	\end{align}
	for some constant~$c\in\mathbb Z_2$ and for some non-empty set~$\mathcal S\subseteq\mathbb Z_m$,
	and if the random variable~$Z$ is uniformly distributed, and~$P_{A,B,Z}=P_{A,B}P_Z$,
	then
	\begin{align}
		\forall P_{A,B}:
		\Pr\left[ 
			A=f(Z)
			\wedge
			B=g(Z)
		\right]
		\leq
		\frac{1}{2}
		\,.
	\end{align}
\end{lemma}
\begin{proof}
	Since~$|\mathcal S| > 1$, the expression~$\bigoplus_{i\in\mathcal S}z_i$
	is a parity function.
	Furthermore, since~$Z$ is uniformly distributed, we have
	\begin{align}
		\Pr\left[f(Z)=g(Z)\right] = \Pr\left[f(Z) \neq g(Z)\right]= \frac 1 2
		\,.
	\end{align}
	Denote by~$p$ the probability that~$A=B$, and note that this probability is {\em independent\/} of the random variable~$Z$.
	We thus have
	\begin{align}
		\Pr&\left[ A=f(Z) \wedge B=g(Z) \right]
		\notag\\
		&
		=
		p
		\Pr\left[ A=f(Z)=g(Z) \mid A= B\right]
		\notag\\
		&
		\quad
		+
		(1-p)
		\Pr\left[ A=f(Z) \neq g(Z) \mid A\neq B\right]
		\\
		&
		\leq
		\frac{1}{2}
		\,.
	\end{align}
\end{proof}

Now, we acquired the tools to prove our first theorem.
\begin{proof}[Proof of Theorem~\ref{thm:bicausalinequality}]
	We start with the analysis on the upper bound of the winning probability of the game~$G_n$, where the bi-causal correlations are restricted: The set~$\mathcal K$ is fixed.
	The parties in the set~$\mathbb Z_n\setminus\mathcal K$ have access to~$X$,~{\it i.e.,}~to all inputs.
	This means that for every~\mbox{$k\in\mathbb Z_n\setminus\mathcal K$}, party~$k$ can guess~$\omega_k^n(X)$ {\em deterministically.}
	Thus, the maximum probability of winning~$G_n$ in this restricted setup equals the maximum probability that the parties in~$\mathcal K$ produce the correct guess:
	\begin{align}
		\begin{split}
			&
			\max_{
				P_{A\mid X}=P_{A_\mathcal K | X_\mathcal K}P_{A_{\setminus\mathcal K}|A_\mathcal K,X}
			}
			\Pr\left[A=\omega^n(X)\right]
			\\
			&\qquad=
			\max_{P_{A_\mathcal K\mid X_\mathcal K}}
			\Pr\left[A_\mathcal K=\omega_\mathcal K^n(X)\right]
			\,.
		\end{split}
	\end{align}
	For every~$P_{A_\mathcal K\mid X_\mathcal K}$, we decompose the winning probability as
	\begin{align}
		\begin{split}
			&\Pr\left[A_\mathcal K=\omega_\mathcal K^n(X)\right]
			\\
			&=
			\Pr\left[A_\mathcal K=\omega_\mathcal K^n(X)
				\mid
				\mathcal E^\mathcal K
			\right]
			\Pr\left[\mathcal E^\mathcal K\right]
			\\
			&
			+
			\Pr\left[A_\mathcal K=\omega_\mathcal K^n(X)
				\mid
				\mathbb Z_2^{|\mathcal K|} \setminus \mathcal E^\mathcal K
			\right]
			\Pr\left[
				\mathbb Z_2^{|\mathcal K|} \setminus \mathcal E^\mathcal K
			\right]
			\,.
		\end{split}
	\end{align}
	First, let us consider the winning probability conditioned on the event~$\mathcal E^\mathcal K$, and let~$k_0\in\mathcal K$ be the only party that has to make a non-trivial guess.
	By Equation~\eqref{eq:revised} and Definition~\ref{def:svetfunction}, we observe that there exists some~$\lambda$ and some function~$f:\mathbb Z_2^{|\mathcal K|}\rightarrow \mathbb Z_2$ such that
	\begin{align}
		\omega_{k_0}^n(x) = S^{n-|\mathcal K|}_{\lambda}(x_{\setminus \mathcal K}) \oplus f(x_\mathcal K)
		\,,
	\end{align}
	where~$S^{n-|\mathcal K|}_{\lambda}$ is an~$(n-|\mathcal K|)$-ary Ardehali-Svetlichny function.
	Thus, for party~$k_0$ it is sufficient to guess~$S^{n-|\mathcal K|}_{\lambda}(X_{\setminus\mathcal K})$ only---the variable~$f(X_\mathcal K)$ is known to party~$k$.
	In combination with Lemma~\ref{lemma:svetlichny}, we obtain
	\begin{align}
		\Pr\left[ 
			A_{k_0} = \omega_{k_0}^n(X)
		\right]
		\leq
		\frac{1}{2} + 2^{-\floor{\frac{n-|\mathcal K|}{2}}-1}
		\,,
	\end{align}
	from which the bound
	\begin{align}
		\Pr\left[A_\mathcal K=\omega_\mathcal K^n(X)
			\mid
			\mathcal E^\mathcal K
		\right]
		\leq
		\frac{1}{2} + 2^{-\floor{\frac{n-|\mathcal K|}{2}}-1}
		\label{eq:pE}
	\end{align}
	follows.
	Note that we can use Lemma~\ref{lemma:svetlichny} in this setting because a single party~$k_0$ cannot guess an Ardehali-Svetlichny function better than multiple parties together.

	In the converse case, {\em at least two\/} parties~\mbox{$k_0,k_1\in\mathcal K$} have to jointly guess the random variables~$\omega_{k_0}^n(X)$ and~$\omega_{k_1}^n(X)$.
	These two functions, however, are parity relative:
	\begin{align}
		\omega_{k_0}^n(x)\oplus\omega_{k_1}^n(x)
		=
		\beta_{k_0}
		\oplus
		\beta_{k_1}
		\oplus
		\bigoplus_{i\not\in\mathcal K}
		x_i\left( \alpha_{k_0,i}\oplus\alpha_{k_1,i} \right)
		\,,
	\end{align}
	where for some~$i\not\in\mathcal K$ the value of~$\alpha_{k_0,i}\oplus\alpha_{k_1,i}$ equals to one.
	Thus, by Lemma~\ref{lemma:parity-relative}, we get
	\begin{align}
		\Pr\left[A_\mathcal K=\omega_\mathcal K^n(X)
			\mid
			\mathbb Z_2^{|\mathcal K|} \setminus \mathcal E^\mathcal K
		\right]
		\leq \frac 1 2
		\,.
		\label{eq:pnotE}
	\end{align}
	
	By Lemma~\ref{lemma:occurenceE}, Equations~\eqref{eq:pE} and~\eqref{eq:pnotE}, the winning probability for bi-causal correlations with a fixed~$\mathcal K$ is
	\begin{align}
		&\Pr\left[A_\mathcal K=\omega_\mathcal K^n(X)\right]
		\notag
		\\
		&
		\leq
		\left( 
		\frac{1}{2} + 2^{-\floor{\frac{n-|\mathcal K|}{2}}-1}
		\right)
		\left(
		2^{-|\mathcal K|+1}
		\right)
		+
		\frac{1-2^{-|\mathcal K|+1}}{2}
		\\
		&=
		\frac 1 2
		+
		2^{-\floor{\frac{n+|\mathcal K|}{2}}}
		\,.
	\end{align}
	Since this holds for every~$\emptyset\subsetneq \mathcal K \subsetneq \mathbb Z_n$, and since bi-causal correlations are arbitrary convex combinations over~$\mathcal K$, we obtain
	\begin{align}
		\Pr&\left[A=\omega^n(X)\right]
		\leq
		\frac 1 2
		+
		2^{-\ceil{n/2}}
	\end{align}
	for all bi-causal correlations.
\end{proof}

\subsection{Saturation of bi-causal inequalities}
\label{sec:saturation}
We inductively show Theorem~\ref{thm:faces} by specifying a bi-causal strategy that reaches the upper bound on the winning probability of Theorem~\ref{thm:bicausalinequality}.
\begin{proof}[Proof of Theorem~\ref{thm:faces}]
	Place party~$0$ in the first subset, {\it i.e.,}~$\mathcal K=\{0\}$, and let party~$0$ deterministically guess the value~$0$.
	In this setting, the winning probability~$\Pr\left[ A=\omega^n(X) \right]$ reduces to the probability~$\Pr\left[ 0=\omega_0^n(X) \right]$.
	As observed above, we can express~$\omega_0^n(x)$ with an Ardehali-Svetlichny function:
	\begin{align}
		\omega_k^n(x) =
			S^{n-1}_{\lambda}(x_{\setminus 0})
			\,,
	\end{align}
	where~$\lambda$ is the alternating sequence~$(0,1,0,\dots)$.
	Party~$0$ guesses correctly if and only if the random variable~$X$ takes a value in the set
	\begin{align}
		\mathcal Z^n :=
		\left\{ 
			x\in\mathbb Z^n_2
			\mid
			S^{n-1}_{\lambda}(x_{\setminus 0}) = 0
		\right\}
		\,.
	\end{align}

	{\it Base case.}
	It is easily verified that~$|\mathcal Z^2|=4$, and~$|\mathcal Z^3|=6$, from which~$\Pr[\mathcal Z^2]=1$, and~$\Pr[\mathcal Z^3]=3/4$ follow.

	{\it Induction step.}
	We take steps of two.
	First, observe that for every~$x\in\mathbb Z_2^{n+2}$ we have
	\begin{align}
		\begin{split}
			S_\lambda^{n+2}(x_{\setminus 0})
			&
			=
			S_\lambda^{n}(x_{\setminus \{0,n,n+1\}})
			\,
			\oplus
			\\
			&
			x_nx_{n+1}
			\oplus
			x_{n(+1)}
			\oplus\\
			&
			\bigoplus_{i\in\mathbb Z_n\setminus \{0\}}
			x_i(x_n\oplus x_{n+1})
			\,,
		\end{split}
	\end{align}
	where the single term~$x_{n(+1)}$ is~$x_n$ if~$n$ is even, and~$x_{n+1}$ otherwise.
	Thus, if~\mbox{$x\in\mathcal Z^n$}, then~$(x,0,0)$,~$(x,1,1)$, and exactly one of~$(x,0,1)$ and~$(x,1,0)$ are in~$\mathcal Z^{n+2}$.
	In the alternative case, if~$x\not\in\mathcal Z^n$, then exactly one of~$(x,0,1)$ and~$(x,1,0)$ is in~$\mathcal Z^{n+2}$.
	Therefore, the cardinalities of these sets are related by
	\begin{align}
		|\mathcal Z^{n+2}| = 2|\mathcal Z^{n}| + 2^n
		\,.
	\end{align}
	By the induction hypothesis~$|\mathcal Z^n|=2^n(1/2+2^{-\ceil{n/2}})$ we therefore obtain
	\begin{align}
		|\mathcal Z^{n+2}|
		&
		=
		2\left( 
		2^n\left( 
		\frac{1}{2}+2^{-\ceil{n/2}}
		\right)
		\right)
		+2^n
		\\
		&
		=
		2^{n+1}
		+
		2^{n+1-\ceil{n/2}}
		\\
		&
		=
		2^{n+2}
		\left( 
		\frac{1}{2}
		+
		2^{-\ceil{(n+2)/2}}
		\right)
		\,,
	\end{align}
	from which~$\Pr[\mathcal Z^{n+2}] = 1/2+2^{-\ceil{(n+2)/2}}$ follows.
	Therefore, for all~$n\geq 2$ there exists at least one bi-causal strategy with which the bi-causal bound is saturated.

	Now, observe that~$\omega^n_0(x)$ is invariant under any relabelling of the parties~$0\leftrightarrow k$ where~$k\in\mathbb Z_n$ is {\em even.}
	Moreover, if we consider an odd number of parties~$n$, then~$\omega^n_0(x)$ is {\em additionally\/} invariant under any relabelling~$0\leftrightarrow \ell$ for~$\ell$ {\em odd.}
	Thus, if~$n$ is even, then there exist at least~$n/2$ bi-causal strategies that saturate the bound, and if~$n$ is odd, then there exist at least~$n$ such strategies.
\end{proof}

\subsection{Deterministic classical violation}
\label{proofs:violation}
We exploit various properties of process functions in order to show Theorem~\ref{thm:lc}.

\subsubsection{Properties of process functions}
Before we present the properties, we introduce element-wise constant functions and their respective reduced functions.
\begin{definition}[Element-wise constant and reduced function]
	An~$n$-ary function~\mbox{$\omega:\bigtimes_{k\in\mathbb Z_n}\mathcal O_k\rightarrow\bigtimes_{k\in\mathbb Z_n}\mathcal I_k$} is {\em element-wise constant\/} if and only if
	\begin{align}
		\forall k\in\mathbb Z_n, o\in\mathcal O,\tilde o_k\in\mathcal O_k:
		\omega_k(o)=\omega_k(o_{\setminus k},\tilde o_k)
		\,.
	\end{align}

	Let~$\omega$ be such a function,
	and let~$f_\ell:\mathcal I_\ell\rightarrow\mathcal O_\ell$ be some function for~$\ell\in\mathbb Z_n$.
	The {\em reduced function\/}~\mbox{$\omega^{f_\ell}: \mathcal O_{\setminus \ell} \rightarrow \mathcal I_{\setminus \ell}$}
	is~$(\omega_0^{f_\ell},\dots,\omega_{\ell-1}^{f_\ell},\omega_{\ell+1}^{f_\ell},\dots,\omega_{n-1}^{f_\ell})$ with
	\begin{align}
		\omega_k^{f_\ell} : o_{\setminus \ell}
		\mapsto
		\omega_k(o_{\setminus \ell},\bar o_\ell)
		\,,
	\end{align}
	where, for some arbitrary~$\tilde o_\ell\in\mathcal O_\ell$,
	\begin{align}
		\bar o_\ell=f_\ell(\omega_\ell(
		o_{\setminus \ell}, \tilde o_\ell))
		\,.
	\end{align}
\end{definition}
\begin{lemma}[Constant~\cite{Baumeler2017,Baumeler2020}]
	\label{lemma:constant}
	If~$\omega$ is an~$n$-party process function, then~$\omega$ is {\em element-wise constant.}
\end{lemma}
We now relate process functions with their reduced functions.
\begin{lemma}[Transitivity~\cite{Baumeler2017}]
	\label{lemma:transitivity}
	Let~$\omega:\mathcal O \rightarrow \mathcal I$ be an~$n$-ary element-wise constant function.
	If there exists some~$k\in\mathbb Z_n$ such that for all~$f_k:\mathcal I_k\rightarrow\mathcal O_k$ the reduced function~$\omega^{f_k}$ is a process function, then~$\omega$ is a process function.
\end{lemma}
The above lemma, proven in Ref.~\cite{Baumeler2017}, can be made stronger: It is sufficient that only for {\em some\/} functions~$f_k$ the reduced function~$\omega^{f_k}$ is a process function.
We show this in the special case where for all~$k$, all sets~$\mathcal O_k$ and~$\mathcal I_k$ are binary.
\begin{lemma}
	\label{lemma:3to4}
	Let~$\omega:\mathbb Z_2^n \rightarrow \mathbb Z_2^n$ be element-wise constant.
	If the reduced functions~$\omega^{f_k}$ for~$f_k$ being the constant-zero, constant-one, and identity function, are process functions,
	then for all~$f_k$,~$\omega^{f_k}$ is a process function.
\end{lemma}
\begin{proof}
	Without loss of generality, and for better presentation, we set~$k$ to~$0$.
	If the reduced functions of~$\omega$ for~$f_{0}$ constant zero, constant one, and identity, are process functions, we have that for all~$f_{\setminus 0}$ there exist {\em unique\/} fixed points for the functions
	\begin{align}
		x &\mapsto \omega_{\setminus 0}\left(0, f_{\setminus 0}(x)\right)\,,\\
		x &\mapsto \omega_{\setminus 0}\left(1, f_{\setminus 0}(x)\right)\,,\\
		x &\mapsto \omega_{\setminus 0}\left(\omega_{0}(0,f_{\setminus 0}(x)),f_{\setminus 0}(x)\right)
		\,.
		\label{eq:wid}
	\end{align}
	Let~$\alpha$ and~$\beta$ be the fixed points of the first two functions,
	{\it i.e.,}~we have the identities
	\begin{align}
		\alpha = \omega_{\setminus 0}\left(0, f_{\setminus 0}(\alpha)\right)\,,\qquad
		\beta = \omega_{\setminus 0}\left(1, f_{\setminus 0}(\beta)\right)
		\,.
	\end{align}
	We define the bits the first party receives upon applying the function~$\omega$ to these fixed points:
	\begin{align}
		\bar \alpha := \omega_{0}\left(0, f_{\setminus 0}(\alpha)\right)\,,\qquad
		\bar \beta := \omega_{0}\left(0, f_{\setminus 0}(\beta)\right)
		\,.
	\end{align}
	Note that these bits are independent of the first argument; the function~$\omega$ is element-wise constant.
	We now show that this implies that the reduced function~$\omega^{f_k}$, where~$f_k$ is the bit-flip function, {\it i.e.,}~the function
	\begin{align}
		x &\mapsto \omega_{\setminus 0}\left(1\oplus \omega_{0}(0,f_{\setminus 0}(x)),f_{\setminus 0}(x)\right)
		\,,
		\label{eq:wnot}
	\end{align}
	also has a fixed point.
	In the case where~$\bar\alpha=1$,~$\alpha$ is a fixed point of Equation~\eqref{eq:wnot}:
	\begin{align}
		\omega_{\setminus 0}\left(1\oplus \omega_{0}(0,f_{\setminus 0}(\alpha)),f_{\setminus 0}(\alpha)\right)
		&=
		\omega_{\setminus 0}\left(0,f_{\setminus 0}(\alpha)\right)\\
		&=
		\alpha
		\,.
	\end{align}
	In the case where~$\bar\beta=0$,~$\beta$ is a fixed point of Equation~\eqref{eq:wnot}:
	\begin{align}
		\omega_{\setminus 0}\left(1\oplus \omega_{0}(0,f_{\setminus 0}(\beta)),f_{\setminus 0}(\beta)\right)
		&=
		\omega_{\setminus 0}\left(1,f_{\setminus 0}(\beta)\right)\\
		&=
		\beta
		\,.
	\end{align}
	The last case, {\it i.e.,}~$\bar\alpha=0$ and~$\bar\beta=1$ {\em cannot\/} arise.
	Assume towards a contradiction that~$\bar\alpha=0$ and~$\bar\beta=1$.
	This implies that~$\alpha$ and~$\beta$ are fixed points of Equation~\eqref{eq:wid}:
	\begin{align}
		\omega_{\setminus 0}\left(\omega_{0}(0,f_{\setminus 0}(\alpha)),f_{\setminus 0}(\alpha)\right)
		&=
		\omega_{\setminus 0}\left(0,f_{\setminus 0}(\alpha)\right)
		=
		\alpha\,,\\
		\omega_{\setminus 0}\left(\omega_{0}(0,f_{\setminus 0}(\beta)),f_{\setminus 0}(\beta)\right)
		&=
		\omega_{\setminus 0}\left(1,f_{\setminus 0}(\beta)\right)
		=
		\beta
		\,.
	\end{align}
	Since Equation~\eqref{eq:wid} has a {\em unique\/} fixed point, we conclude~$\alpha=\beta$ which contradicts~$\bar\alpha\not=\bar\beta$.
	By noting that there is a total of four functions from a bit to a bit~({\it i.e.,}~constant zero, constant one, identity, and bit-flip), the proof is concluded.
\end{proof}

\subsubsection{Proof of Theorem~\ref{thm:lc}}
The above properties allow us to prove our statement: The functions~$\omega^n$ of the game~$G_n$ are process functions.
A schematic representation of the proof is given in Figure~\ref{fig:reductions}.
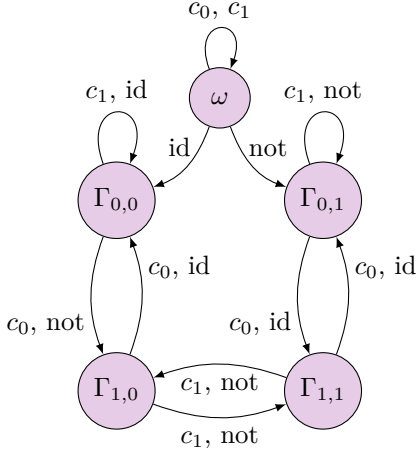
\begin{figure}
	\centering
	\begin{tikzpicture}[>=latex,pf/.style={circle,draw,fill=violet!20!white,minimum size=0.8cm,text centered}]
		\node[pf] (p0) {$\omega$};
		\node[pf,below left=1cm of p0] (p1) {$\Gamma_{0,0}$};
		\node[pf,below right=1cm of p0] (p2) {$\Gamma_{0,1}$};
		\node[pf,below=1.5cm of p1] (p3) {$\Gamma_{1,0}$};
		\node[pf,below=1.5cm of p2] (p4) {$\Gamma_{1,1}$};
		\draw[->] (p0) to[out=110,in=70,looseness=7] node[midway,above] {$c_0$, $c_1$} (p0);
		\draw[->] (p1) to[out=110,in=70,looseness=7] node[midway,above] {$c_1$, id} (p1);
		\draw[->] (p2) to[out=110,in=70,looseness=7] node[midway,above] {$c_1$, not} (p2);
		\draw[->] (p0) to[out=250,in=20] node[near start,left] {id} (p1);
		\draw[->] (p0) to[out=290,in=160] node[near start,right] {not} (p2);
		\draw[->] (p1) to[out=250,in=110] node[near end,left] {$c_0$, not} (p3);
		\draw[->] (p3) to[out=70,in=290] node[near end,right] {$c_0$, id} (p1);
		\draw[->] (p2) to[out=250,in=110] node[near end,left] {$c_0$, id} (p4);
		\draw[->] (p4) to[out=70,in=290] node[near end,right] {$c_0$, id} (p2);
		\draw[->] (p3) to[out=-20,in=200] node[midway,below] {$c_1$, not} (p4);
		\draw[->] (p4) to[out=160,in=20] node[midway,below] {$c_1$, not} (p3);
	\end{tikzpicture}
	\caption{Schematic representation of the proof of Theorem~\ref{thm:lc}.
	The arrows point from~$n$-party to the respective~$(n-1)$-party functions where the last party applies the identity (id), not, constant-zero ($c_0$), or constant-one ($c_1$) function.
	Due to Lemma~\ref{lemma:transitivity} and the fact that~$\omega^3$ and~$\Gamma^3_{\alpha,\beta}$ for all~$\alpha,\beta\in\{0,1\}$ are process functions (see Appendix~\ref{app:gamma}), we establish that~$\omega^n$ is a process function for all~$n$.
	Lemma~\ref{lemma:3to4} allows us to neglect the edges labeled by not.
}
	\label{fig:reductions}
\end{figure}
\begin{proof}[Proof of Theorem~\ref{thm:lc}]
	We show this theorem by induction.
	Assume towards a contradiction that the function~$\omega^n$ is {\em not\/} a process function.
	Clearly,~$\omega^n$ is element-wise constant.
	Thus, by Lemma~\ref{lemma:transitivity} there exists some function~$f_{n-1}:\mathbb Z_2\rightarrow\mathbb Z_2$ such that the reduced function~$\omega^{n,{f_{n-1}}}$, {\it i.e.,}~the function~$\omega^n$ where party~$n-1$ implements~$f_{n-1}$, is {\em not\/} a process function.
	Furthermore, by Lemma~\ref{lemma:3to4} we know that this will be the case where~$f_{n-1}$ is the constant-zero, constant-one, or identity function.

	{\em Constant-zero function.}
	In the case where~$f_{n-1}$ is the constant-zero function, the reduced function~$\omega^{n,{f_{n-1}}}$ equals~$\omega^{n-1}$.
	To show this, we express~$\omega_k^n(x)$ in terms of~$\omega_k^{n-1}(x_{\setminus n-1})$:
	\begin{align}
		\begin{split}
			\omega_k^n(x) 
			&
			=
			\omega_k^{n-1}(x_{\setminus n-1})
			\oplus 
			\bigoplus_{i\in\mathbb Z_{n-1}\setminus\{k\}}
			x_ix_{n-1}
			\\
			&
			\oplus
			\left[
			k\not\equiv_2 n
			\right]
			x_{n-1}
			\label{eq:intermsofless}
			\,,
		\end{split}
	\end{align}
	where we use the notation that~$[p]=1$ if~$p$ holds, and~$0$ otherwise.
	Now, for all~$k\in\mathbb Z_{n-1}$, and where~$f_{n-1}$ is the constant-zero function, we have
	\begin{align}
		\omega_k^{n,f_{n-1}}(x)
		= \omega_k^n(x_{\setminus n-1},0) 
		= \omega_k^{n-1}(x_{\setminus n-1})
		\,.
	\end{align}

	{\em Constant-one function.}
	In the case where~$f_{n-1}$ is the constant-one function, the reduced function~$\omega^{n,{f_{n-1}}}$ is equivalent to~$\omega^{n-1}$ up to a permutation of the parties and a constant.
	By Equation~\eqref{eq:intermsofless} we get
	\begin{align}
		\omega_k^{n,f_{n-1}}(x)
		&=
		\omega_k^n(x_{\setminus n-1},1) 
		\\
		&=
		\omega_k^{n-1}(x_{\setminus n-1})
		\oplus
		\!\!\!\bigoplus_{i\in\mathbb Z_{n-1}\setminus\{k\}}\!\!\!
		x_i
		\oplus
		\left[
		k\not\equiv_2 n
		\right] \\
		&=
		\bar{\omega}_k^{n-1}(x_{\setminus n-1})
		\oplus
		\left[
		k\not\equiv_2 n
		\right]  
		\,.
	\end{align}
	The function~$\bar{\omega}^{n-1}$ is equal to the function~$\omega^{n-1}$ where party~$k\in\mathbb Z_{n-1}$ becomes party~$n-2-k$.
	This fact is shown in Appendix~\ref{app:equivalence}.
	Hence,~$\omega^{n-1}$ is a process function if and only if~$\bar{\omega}^{n-1}$ is a process function.
	The additional term~$[k\not\equiv_2 n]$ is independent of the input to the function and can therefore be simulated by any party.
	This means that, since we assumed~$\omega^n$ {\em not\/} to be a process function, at least one of~$\omega^{n-1}$ or~$\omega^{n,f_{n-1}}$, where~$f_{n-1}$ is the identity function, is also not a process function.

	{\em Identity function.}
	For the last case, let~$f_{n-1}$ be the identity function.
	In Appendix~\ref{app:id}, we show that the reduced function, in this case, is~$\Gamma_{0,0}^{n-1}$.
	Appendix~\ref{app:gamma} shows that~$\Gamma_{0,0}^{n-1}$ is a process function.

	Since~$\omega^3$ (see Equation~\eqref{eq:threepartygame}) and~$\Gamma_{0,0}^{n-1}$ are process functions, we reach a contradiction:~$\omega^n$ is a process function.
\end{proof}

\subsection{All process functions are process matrices}
\label{proofs:quantum}
We prove Theorem~\ref{thm:qc} constructively with the limit theorem~\cite{Baumeler2016,Baumelerphd}.
\begin{proof}[Proof of Theorem~\ref{thm:qc}]
	Let~$\omega:\bigtimes_{k\in\mathbb Z_n} \mathcal O_k \rightarrow \bigtimes_{k\in\mathbb Z_n} \mathcal I_k$ be a process function and define the operator
	\begin{align}
		W^\omega:=\sum_{o\in\mathcal O}
		\ket o\bra o_O
		\otimes
		\ket{\omega(o)}\bra{\omega(o)}_I
	\end{align}
	with appropriate Hilbert spaces,
	and where~$\{\ket o\}_{o\in\mathcal O}$ and~$\{\ket i\}_{i\in\mathcal I}$ is a fixed basis for these Hilbert spaces.
	The limit theorem states that for every collection of completely positive trace-preserving maps~\mbox{$(\mu_k:\mathcal L(\mathcal I_k)\rightarrow\mathcal L(\mathcal O_k))_{k\in\mathbb Z_n}$}, where~$M_k$ are the corresponding Choi operators and~$M=\bigotimes_{k\in\mathbb Z_n}M_k$ is diagonal in the same basis as~$W^\omega$:
	\begin{align}
		\Tr[MW^\omega] = 1
		\,.
	\end{align}
	Now let~$M'=\bigotimes_{k\in\mathbb Z_n}M'_k$ be the Choi operator of the completely positive trace-preserving maps~\mbox{$(\mu'_k:\mathcal L(\mathcal I_k)\rightarrow\mathcal L(\mathcal O_k))_{k\in\mathbb Z_n}$} not necessarily diagonal in the same basis as~$W^\omega$.
	Then
	\begin{align}
		\Tr[M'W^\omega] = \Tr[M'_\text{diag}W^\omega]
	\end{align}
	with
	\begin{align}
		M'_\text{diag} :=
		\!\!
		\sum_{(o,i)\in\mathcal O\times\mathcal I}
		\!\!
		\ket{o,i}
		\bra{o,i}
		M'
		\ket{o,i}
		\bra{o,i}
		\,.
	\end{align}
	Since~$M'_\text{diag}$ is the Choi operator of a completely positive trace-preserving map~($M'_\text{diag} \geq 0, \Tr_{O} M_\text{diag} = \id$)
	and diagonal in the same basis as~$W^\omega$, we have
	\begin{align}
		\Tr[M'W^\omega] = \Tr[M'_\text{diag}W^\omega] = 1
		\,.
	\end{align}
	We conclude the proof by noting that~$W^\omega$ is positive semi-definite and that it implements the same dynamics as~$\omega$.
\end{proof}

\noindent
{\bf Acknowledgments.}
We thank Alastair Abbott, Costantino Budroni, Fabio Costa, Paul Erker, Simon Milz, and Eleftherios Tselentis for enlightening discussion.
We thank two anonymous reviewers for their helpful comments and insights.
\"AB~is supported by the Austrian Science Fund~(FWF) through projects ZK3 (Zukunftskolleg) and F7103 (BeyondC), and by the Erwin Schr\"odinger Center for Quantum Science~\&~Technology (ESQ).

\bibliographystyle{quantum}
\bibliography{refs}
\onecolumn
\newpage
\appendix

\section{Permutation of the parties}
\label{app:equivalence}
We define the function~$\bar{\omega}^n:\mathbb Z_2^n \rightarrow \mathbb Z_2^n$ as~$\bar{\omega}^n(x)=(\bar{\omega}^n_0(x),\bar{\omega}^n_1(x),\dots,\bar{\omega}^n_{n-1}(x))$, where for all~$k\in\mathbb Z_n$ we have
\begin{align}
	\begin{split}
		\bar{\omega}_k^n: \mathbb Z_2^n &\rightarrow \mathbb Z_2\\
		x&\mapsto
		\bigoplus_{\substack{i,j\in\mathbb Z_n\setminus\{k\}\\i < j}}
		x_ix_j
		\oplus
		\bigoplus_{i\in\mathbb Z_n}
		\bar{\gamma}_{k,i} x_i
		\,,
	\end{split}
\end{align}
with
\begin{align}
	\bar{\gamma}_{k,i} :=
	\begin{cases}
		1 & (i < k \wedge i\equiv_2 k ) \vee (k < i \wedge i\not\equiv_2 k)\\
		0 & \text{otherwise.}
	\end{cases}
\end{align}
In comparison to~$\omega^n$, this function uses the alternative single terms~({\it cf.\/}~Equation~\eqref{eq:chi}).

\begin{lemma}
	The function~$\bar{\omega}^{n}$ is equivalent to~$\omega^n$ under a relabeling of the parties.
\end{lemma}

\begin{proof}
	Reversing the order of parties in~$\omega^n$ gives~$\bar{\omega}^n$, {\it i.e.,}~if party~$k$ in the function~$\omega^n$ becomes party~$n-k-1$, then we obtain~$\bar{\omega}^n$.
	We denote with~$\omega^{n \circlearrowright}_k$ the function of the $k$-th party where the order of parties in $\omega^n_k$ is reversed.
	More precisely, the input to and the output of the $k$-th party is $x^\circlearrowright_k = x_{n-k-1}$ and $\omega^n_{n-k-1} (x^\circlearrowright)$ where $x^\circlearrowright$ represents the reversed input string. 
	This equivalence is shown in the following calculation, where we use~$i':= n-i-1$, and~$j':= n-j-1$:
	\begin{align}
		\omega^{n \circlearrowright}_k (x)
		&=
		\omega^n_{n-k-1} (x^\circlearrowright)
		\\
		&=
		\bigoplus_{\substack{i,j\in\mathbb Z_n\setminus\{n-k-1\}\\i < j}}
		x_{n-i-1}x_{n-i-1}
		\oplus
		\bigoplus_{i\in\mathbb Z_n}
		\gamma_{n-k-1,i}
		x_{n-i-1}
		\\
		&=
		\bigoplus_{\substack{i',j'\in\mathbb Z_n\setminus\{k\}\\i' < j'}}
		x_{i'}x_{j'}
		\oplus
		\bigoplus_{i'\in\mathbb Z_n}
		\bar\gamma_{k,i'}
		x_{i'}
		\\
		&=
		\bar{\omega}^n_k(x)
		\,.
	\end{align}
\end{proof}

\section{Reduced function of~$\omega^n$ where party~$n-1$ implements the identity}
\label{app:id}
Let the~$f_{n-1}$ be the identity function.
The reduced function in this case is
\begin{align}
	\omega_k^{n,f_{n-1}}(x)
	&=
	\omega_k^n(x_{\setminus n-1},\omega_{n-1}^n(x_{\setminus n-1},0)) 
	\,.
	\label{eq:svetid}
\end{align}
The input to party~$n-1$, {\it i.e.,}~$\omega_{n-1}^n(x_{\setminus n-1},0)$, is
\begin{align}
	\bigoplus_{\substack{i,j\in\mathbb Z_{n-1}\\i<j}}
	x_ix_j
	\oplus
	\bigoplus_{\substack{i\in\mathbb Z_{n-1}\\i\equiv_2 n}}
	x_i
	=
	\bigoplus_{\substack{i,j\in\mathbb Z_{n-1}\setminus\{k\}\\i<j}}
	x_ix_j
	\oplus
	\bigoplus_{i\in\mathbb Z_{n-1}\setminus\{k\}}
	x_ix_k
	\oplus
	\bigoplus_{\substack{i\in\mathbb Z_{n-1}\setminus\{k\}\\i\equiv_2 n}}
	x_i
	\oplus
	\left[
		k \equiv_2 n
	\right]
	x_k
	\,.
\end{align}
This expression is now plugged into Equation~\eqref{eq:svetid}, {\it i.e.,}~the variable~$x_{n-1}$ in Equation~\eqref{eq:intermsofless} takes the above value.
We evaluate the two expressions in Equation~\eqref{eq:intermsofless} involving~$x_{n-1}$ separately.
The first term is
\begin{align}
	\begin{split}
		&\bigoplus_{i\in\mathbb Z_{n-1}\setminus\{k\}}
		x_i\omega_{n-1}^n(x_{\setminus n-1},0)
		\\
		&
		=
		\bigoplus_{\substack{i,j,\ell\in\mathbb Z_{n-1}\setminus\{k\}\\i<j<\ell}}
		x_ix_jx_\ell
		\oplus
		\bigoplus_{i\in\mathbb Z_{n-1}\setminus\{k\}}
		x_ix_k
		\oplus
		\bigoplus_{\substack{i,j\in\mathbb Z_{n-1}\setminus\{k\}\\i<j\\i\not\equiv_2 j}}
		x_ix_j
		\oplus
		\bigoplus_{\substack{i\in\mathbb Z_{n-1}\setminus\{k\}\\i\equiv_2 n}}
		x_i
		\oplus
		\left[
			k\equiv_2 n
		\right]
		\bigoplus_{i\in\mathbb Z_{n-1}\setminus \{k\}}
		x_ix_k
		\,,
	\end{split}
\end{align}
and the second is
\begin{align}
	\left[
		k\not\equiv_2 n
	\right]
	\omega_{n-1}^n(x_{\setminus n-1},0)
	=
	\left[
		k\not\equiv_2 n
	\right]
	\left(
	\bigoplus_{\substack{i,j\in\mathbb Z_{n-1}\setminus\{k\}\\i<j}}
	x_ix_j
	\oplus
	\bigoplus_{i\in\mathbb Z_{n-1}\setminus\{k\}}
	x_ix_k
	\oplus
	\bigoplus_{\substack{i\in\mathbb Z_{n-1}\setminus\{k\}\\i\equiv_2 n}}
	x_i
	\right)
	\,.
\end{align}
Thus, for all~$k\in\mathbb Z_{n-1}$, Equation~\eqref{eq:svetid} is the parity of~$\omega_k^{n-1}(x_{\setminus n-1})$ and these last two expressions:
\begin{align}
	\begin{split}
		\omega_k^{n,f_{n-1}}(x)
		&
		=
		\bigoplus_{\substack{i,j,\ell\in\mathbb Z_{n-1}\setminus\{k\}\\i<j<\ell}}
		x_ix_jx_\ell
		\oplus
		\bigoplus_{\substack{i,j\in\mathbb Z_{n-1}\setminus\{k\}\\i<j\\i\not\equiv_2 j}}
		x_ix_j
		\oplus
		\left[ k\equiv_2 n \right]
		\bigoplus_{\substack{i,j\in\mathbb Z_{n-1}\setminus\{k\}\\i<j}}
		x_ix_j
		\\
		&
		\qquad\qquad
		\oplus
		\bigoplus_{i\in\mathbb Z_{n-1}}
		\gamma_{k,i}
		x_i
		\oplus
		\left[ k\equiv_2 n \right]
		\bigoplus_{\substack{i\in\mathbb Z_{n-1}\setminus \{k\}\\i\equiv_2 n}}
		x_i
		\,.
	\end{split}
\end{align}
Therefore, the function~$\omega_k^{n,f_{n-1}}$, where~$f_{n-1}$ is the identity function for party~$n-1$, equals the function~$\Gamma_{0,0,k}^{n-1}$ defined in Appendix~\ref{app:gamma}.
That same appendix also shows that~$\Gamma_{0,0}^{n-1}$ is a process function.

\section{Family of process functions}
\label{app:gamma}
For all~$\alpha,\beta\in\{0,1\}$, we define the function~$\Gamma_{\alpha, \beta}^n:\mathbb Z_2^n \rightarrow \mathbb Z_2^n$ as~$\Gamma_{\alpha,\beta}^n(x)= (\Gamma_{\alpha,\beta,0}^n(x),\Gamma_{\alpha,\beta,1}^n(x),\dots,\Gamma_{\alpha,\beta,n-1}^n(x))$, where for all~$k\in\mathbb Z_n$
\begin{align}
	\begin{split}
		\Gamma_{\alpha,\beta,k}^n : \mathbb Z_2^n &\rightarrow \mathbb Z_2\\
		x&\mapsto
		\bigoplus_{\substack{i,j,\ell\in\mathbb Z_n\setminus\{k\}\\i<j<\ell}}
		x_ix_jx_\ell
		\oplus
		\bigoplus_{\substack{i,j\in\mathbb Z_n\setminus\{k\}\\i<j\\i\not\equiv_2 j}}
		x_ix_j
		\oplus
		\left[ \alpha \not\equiv_2 k+n \right]
		\bigoplus_{\substack{i,j\in\mathbb Z_n\setminus \{k\}\\i<j}}
		x_ix_j
		\\
		&\qquad\qquad
		\oplus
		\bigoplus_{i\in\mathbb Z_n}
		\left(
		(\beta\oplus 1)
		\gamma_{k,i}
		+
		\beta\bar\gamma_{k,i}
		\right)
		x_i
		\oplus
		\left[ \alpha\not\equiv_2 k+n \right]
		\bigoplus_{\substack{i\in\mathbb Z_n\setminus\{k\}\\i\equiv_2 k}}
		x_i
		\,.
	\end{split}
\end{align}
\begin{theorem}
	The function~$\Gamma_{\alpha,\beta}^n$ for all~$\alpha, \beta\in\{0,1\}$ and for all~$n\geq 1$ is a process function.
\end{theorem}
\begin{proof}
	We explicitly prove this statement for~$\alpha=\beta=0$; the other cases are analogous.
	For better presentation, define~$\Gamma_k^n:=\Gamma_{0,0,k}^n$ and~$\Gamma^n:=\Gamma_{0,0}^n$.
	The proof idea is the same as for Theorem~\ref{thm:lc}.
	First, we assume that~$\Gamma^n$ is {\em not\/} a process function.
	Since this function is element-wise constant, it follows that at least one of the reduced functions~$\Gamma^{n,f_{n-1}}$, where~$f_{n-1}$ is the constant-zero, constant-one, or identity function, is not a process function~(see~Lemma~\ref{lemma:transitivity} and Lemma~\ref{lemma:3to4}).
	First, we express~$\Gamma_k^n$ in terms of~$\Gamma_k^{n-1}$:
	\begin{align}
		\begin{split}
			\Gamma_k^n(x)
			&= 
			\Gamma_k^{n-1}(x_{\setminus n-1})
			\oplus
			\bigoplus_{\substack{i,j\in\mathbb Z_{n-1}\setminus\{k\}\\i<j}}
			x_ix_jx_{n-1}
			\oplus
			\bigoplus_{\substack{i\in\mathbb Z_{n-1}\setminus\{k\}\\i\equiv_2 n}}
			x_ix_{n-1}
			\oplus
			\left[ k\not\equiv_2 n \right]
			\bigoplus_{i\in\mathbb Z_{n-1}\setminus\{k\}}
			x_i
			x_{n-1}
			\\
			&\qquad\qquad
			\oplus
			\bigoplus_{\substack{i,j\in\mathbb Z_{n-1}\setminus\{k\}\\i<j}}
			x_ix_j
			\oplus
			\bigoplus_{\substack{i\in\mathbb Z_{n-1}\setminus\{k\}\\i\equiv_2 k}}
			x_i
			\, .
			\label{eq:gammared}
		\end{split}
	\end{align}
	In the case where~$f_{n-1}$ is the constant-zero function, we obtain for the reduced function
	\begin{align}
		\Gamma_k^{n,f_{n-1}}(x)
		&
		=
		\Gamma_k^n(x,0)
		=
	        \Gamma_k^{n-1}(x)
		\oplus
		\bigoplus_{\substack{i,j\in\mathbb Z_{n-1}\setminus\{k\}\\i<j}}
		x_ix_j
		\oplus
		\bigoplus_{\substack{i\in\mathbb Z_{n-1}\setminus\{k\}\\i\equiv_2 k}}
		x_i
		=
		\Gamma_{1,0,k}^{n-1}(x)
		\,.
	\end{align}
	In the case where~$f_{n-1}$ is the constant-one function, we obtain for the reduced function
	\begin{align}
		\Gamma_k^{n,f_{n-1}}(x)
		&
		=
		\Gamma_k^n(x,1)
		=
		\Gamma_k^{n-1}(x)
		\oplus
		\bigoplus_{\substack{i\in\mathbb Z_{n-1}\setminus \{k\}\\i\equiv_2 k}}
		x_i
		\oplus 
		\bigoplus_{\substack{i\in\mathbb Z_{n-1}\setminus \{k\}\\i\equiv_2 n}}
		x_i
		\oplus 
		\left[
		k \not\equiv_2 n 
		\right]
		\bigoplus_{i\in\mathbb Z_{n-1}\setminus \{k\}}
		x_i
		=
		\Gamma^{n-1}_k(x)
		\,.
	\end{align}
	Finally, let~$f_{n-1}$ be the identity function.
	In this case~$\Gamma^{n,f_{n-1}}$ is
	\begin{align}
		\Gamma^{n,f_{n-1}}\left(x_{\setminus n-1}\right) = \Gamma_{\setminus n-1}^n\left(x_{\setminus n-1}, \Gamma_{n-1}^n(x)\right)
		\,.
		\label{eq:gammawithid}
	\end{align}
	We first express~$\Gamma_{n-1}^n(x)$:
	\begin{align}
		\Gamma_{n-1}^n(x) = 
		\bigoplus_{\substack{i,j,\ell\in\mathbb Z_{n-1}\\i<j<\ell}}
		x_ix_jx_\ell
		\oplus
		\bigoplus_{\substack{i,j\in\mathbb Z_{n-1}\\i<j\\i\equiv_2 j}}
		x_ix_j
		\oplus
		\bigoplus_{i\in\mathbb Z_{n-1}}
		x_i
		\,.
	\end{align}
	Now, we evaluate the terms in~$\Gamma^n(x)$ that involve~$x_{n-1}$ (see Equation~\eqref{eq:gammared})---the term~$x_{n-1}$ is replaced by~$\Gamma_{n-1}^n(x)$.
	The first term involving~$x_{n-1}$ becomes
	\begin{align}
		\bigoplus_{\substack{i,j\in\mathbb Z_{n-1}\setminus\{k\}\\i<j}}
		x_ix_j\Gamma_{n-1}^n(x)
		\,.
		\label{eq:gammafirstproduct}
	\end{align}
	Let us perform this calculation step by step.
	By taking the product with the first term in~$\Gamma_{n-1}^n(x)$, we obtain
	\begin{align}
		\bigoplus_{\substack{i,j\in\mathbb Z_{n-1}\setminus\{k\}\\i<j}}
		x_ix_j
		\left(
		\bigoplus_{\substack{i,j,\ell\in\mathbb Z_{n-1}\\i<j<\ell}}
		x_ix_jx_\ell
		\right)
		=
		\bigoplus_{\substack{i,j,\ell\in\mathbb Z_{n-1}\setminus\{k\}\\i<j<\ell}}
		x_ix_jx_\ell
		\oplus
		x_k
		\left( 
		\bigoplus_{\substack{i,j\in\mathbb Z_{n-1}\setminus\{k\}\\i<j}}
		x_ix_j
		\right)
		\,.
	\end{align}
	Then again, the product with the second term in~$\Gamma_{n-1}^n(x)$, gives
	\begin{align}
		\begin{split}
			&
			\bigoplus_{\substack{i,j\in\mathbb Z_{n-1}\setminus\{k\}\\i<j}}
			x_ix_j
			\left(
			\bigoplus_{\substack{i,j\in\mathbb Z_{n-1}\\i<j\\i\equiv_2 j}}
			x_ix_j
			\right)
			\\
			&
			=
			\bigoplus_{\substack{i,j,\ell,m\in\mathbb Z_{n-1}\setminus\{k\}\\i<j<\ell<m\\i+j+\ell+m\equiv_2 1}}
			x_ix_jx_\ell x_m
			\oplus
			x_k
			\left( 
			\bigoplus_{\substack{i,j,\ell\in\mathbb Z_{n-1}\setminus\{k\}\\i<j<\ell\\i+j+\ell\equiv_2 k}}
			x_ix_jx_\ell
			\right)
			\oplus
			x_k
			\left( 
			\bigoplus_{\substack{i,j\in\mathbb Z_{n-1}\setminus\{k\}\\i<j\\i\not\equiv_2 j}}
			x_ix_j
			\right)
			\oplus
			\left( 
			\bigoplus_{\substack{i,j\in\mathbb Z_{n-1}\setminus\{k\}\\i<j\\i\equiv_2 j}}
			x_ix_j
			\right)
			\,.
		\end{split}
	\end{align}
	Finally, the product with the third term in~$\Gamma_{n-1}^n(x)$ is
	\begin{align}
		\bigoplus_{\substack{i,j\in\mathbb Z_{n-1}\setminus\{k\}\\i<j}}
		x_ix_j
		\left(
		\bigoplus_{i\in\mathbb Z_{n-1}}
		x_i
		\right)
		=
		\bigoplus_{\substack{i,j,\ell\in\mathbb Z_{n-1}\\i<j<\ell}}
		x_ix_jx_\ell
		\,.
	\end{align}
	By taking the sum modulo two of the above three expressions, we obtain that Equation~\eqref{eq:gammafirstproduct} equals
	\begin{align}
		\bigoplus_{\substack{i,j,\ell,m\in\mathbb Z_{n-1}\setminus\{k\}\\i<j<\ell<m\\i+j+\ell+m\equiv_2 1}}
		x_ix_jx_\ell x_m
		\oplus
		\bigoplus_{\substack{i,j,\ell\in\mathbb Z_{n-1}\setminus\{k\}\\i<j<\ell\\i+j+\ell\equiv_2 k}}
		x_ix_jx_\ell x_k
		\oplus
		\bigoplus_{\substack{i,j\in\mathbb Z_{n-1}\setminus\{k\}\\i<j\\i+j\equiv_2 1}}
		x_ix_jx_k
		\oplus
		\bigoplus_{\substack{i,j\in\mathbb Z_{n-1}\setminus\{k\}\\i<j\\i+j\equiv_2 0}}
		x_ix_j
		\,.
	\end{align}
	By similar calculations, we replace the variable~$x_{n-1}$ appearing in the second and third term of Equation~\eqref{eq:gammared} with~$\Gamma_{n-1}^n(x)$, and obtain the expressions
	\begin{align}
		\begin{split}
			\bigoplus_{\substack{i\in\mathbb Z_{n-1}\setminus\{k\}\\i\equiv_2 n}}
			&
			x_i\Gamma_{n-1}^n(x)
			=
			\bigoplus_{\substack{i,j,\ell,m\in\mathbb Z_{n-1}\setminus\{k\}\\i<j<\ell<m\\i+j+\ell+m\equiv_2 1}}
			x_ix_jx_\ell x_m
			\oplus
			\bigoplus_{\substack{i,j,\ell\in\mathbb Z_{n-1}\setminus\{k\}\\i<j<\ell\\i+j+\ell\equiv_2 n}}
			x_ix_jx_\ell x_k
			\\
			\oplus
			&
			\left[
				k\equiv_2 n
			\right]
			\bigoplus_{\substack{i,j\in\mathbb Z_{n-1}\setminus\{k\}\\i<j\\i+j\equiv_2 1}}
			x_ix_jx_k
			\oplus
			\bigoplus_{\substack{i,j\in\mathbb Z_{n-1}\setminus\{k\}\\i<j\\i+j\equiv_2 1}}
			x_ix_j
			\oplus
			\left[
				k\not\equiv_2 n
			\right]
			\bigoplus_{\substack{i\in\mathbb Z_{n-1}\setminus\{k\}\\i\equiv_2 n}}
			x_ix_k
			\oplus
			\bigoplus_{\substack{i\in\mathbb Z_{n-1}\setminus\{k\}\\i\equiv_2 n}}
			x_i
			\,,
		\end{split}
	\end{align}
	and
	\begin{align}
		\begin{split}
			\left[
				k\not\equiv_2 n
			\right]
			\bigoplus_{i\in\mathbb Z_{n-1}\setminus\{k\}}
			&
			x_i\Gamma_{n-1}^n(x)
			=
			\left[
				k\not\equiv_2 n
			\right]
			\left(
			\bigoplus_{\substack{i,j,\ell\in\mathbb Z_{n-1}\setminus\{k\}\\i<j<\ell}}
			x_ix_jx_\ell x_k
			\oplus
			\bigoplus_{\substack{i,j\in\mathbb Z_{n-1}\setminus\{k\}\\i<j\\i+j\equiv_2 1}}
			x_ix_jx_k
			\right.
			\\
			&
			\left.
			\oplus
			\bigoplus_{\substack{i\in\mathbb Z_{n-1}\setminus\{k\}\\i\equiv_2 n}}
			x_ix_k
			\oplus
			\bigoplus_{i\in\mathbb Z_{n-1}\setminus\{k\}}
			x_i
			\right)
			\,.
		\end{split}
	\end{align}
	Now, we have everything at hand to evaluate Equation~\eqref{eq:gammawithid}:
	The resulting function is~$\Gamma_k^{n-1}$.
	
	Hence, if~$f_{n-1}$ is the constant-zero function, we retrieve~$\Gamma_{1,0,k}^{n-1}$, if~$f_{n-1}$ is the constant-one or the identity function, we retrieve~$\Gamma_k^{n-1}$.
	For the other values of~$\alpha$ and~$\beta$ the proof is similar.
	The reduced functions are summarized in Table~\ref{tab:red}, we also include---simply for completeness---the case where~$f_{n-1}$ is the bit-flip function.
	\begin{table}
		\begin{center}
			\begin{tabular}{|c|c|c|c|c|} 
				\hline
				& Constant $0$ & Constant $1$ & Identity & Negation\\ 
				\hline
				$\Gamma_{0,0,k}^n$ & 
				$\Gamma_{1,0,k}^{n-1}$ & 
				$\Gamma_{0,0,k}^{n-1}$ & 
				$\Gamma_{0,0,k}^{n-1}$ & 
				$\Gamma_{1,0,k}^{n-1}$\\ 
				\hline
				$\Gamma_{0,1,k}^{n}$ & 
				$\Gamma_{1,1,k}^{n-1}$ & 
				$\Gamma_{0,1,k}^{n-1} \oplus 1$ & 
				$\Gamma_{1,1,k}^{n-1}$ & 
				$\Gamma_{0,1,k}^{n-1} \oplus 1$\\ 
				\hline
				$\Gamma_{1,0,k}^{n}$ & 
				$\Gamma_{0,0,k}^{n-1}$ & 
				$\Gamma_{1,1,k}^{n-1} \oplus \left[k\neq n\right]$ & 
				$\Gamma_{0,0,k}^{n-1}$ & 
				$\Gamma_{1,1,k}^{n-1} \oplus \left[k\neq n\right]$\\ 
				\hline
				$\Gamma_{1,1,k}^{n}$ & 
				$\Gamma_{0,1,k}^{n-1}$ & 
				$\Gamma_{1,0,k}^{n-1} \oplus \left[k=n\right]$ & 
				$\Gamma_{0,1,k}^{n-1}$ & 
				$\Gamma_{1,0,k}^{n-1} \oplus \left[k=n\right]$\\
				\hline
			\end{tabular}
		\end{center}
		\caption{
			\label{tab:red}
			The reduced function of~$\Gamma_{\alpha,\beta}^{n}$ for all functions~$f_{n-1}$ of party~$n-1$.
		}
	\end{table}

	What remains to show is the base case.
	For~$n=3$, we get that~$\Gamma_{1,\beta}^3$ is the process function described in Equation~\eqref{eq:threepartygame}, and thus is yet another generalization of that three-party function.
	Then again,~$\Gamma_{0,\beta}^3$ is {\em causal\/} and therefore a process function as well.
\end{proof}

\end{document}